\documentclass[twoside]{article}

\usepackage{algorithm}
\usepackage[noend]{algorithmic}
\usepackage{amsmath}
\usepackage{amssymb}
\usepackage{amsthm}
\usepackage{parskip}
\usepackage{amsmath} 
\usepackage[hidelinks]{hyperref}

\newtheorem{proposition}{Proposition}
\newtheorem{lemma}{Lemma}
\newtheorem{theorem}{Theorem}

\theoremstyle{definition}
\newtheorem{fact}{Fact}
\newtheorem{assumption}{Assumption}

\newcommand{\uni}{\mathcal{U}}
\newcommand{\epsi}{\varepsilon}
\newcommand{\marge}[2]{ \Delta \left( #1 | #2 \right) }
\newcommand{\opt}{\textsc{OPT}}
\newcommand{\oh}[1]{\mathcal{O}\left( #1 \right)}

\DeclareMathOperator*{\argmax}{arg\,max}
\DeclareMathOperator*{\polylog}{polylog}

\newcommand{\ie}{\textit{i.e.}\xspace}
\newcommand{\eg}{\textit{e.g.}\xspace}

\usepackage{xspace}

\newcommand{\qs}{\textsc{QuickPrune}\xspace}
\newcommand{\qss}{\textsc{QuickPrune-Single}\xspace}
\newcommand{\sm}{\textsc{SM-K}\xspace}
\newcommand{\nhi}{\textsc{NHI}\xspace}
\newcommand{\gcomb}{\textsc{GCOMB-P}\xspace}
\newcommand{\lense}{\textsc{LeNSE}\xspace}
\newcommand{\comb}{\textsc{COMBHelper}\xspace}
\newcommand{\gnn}{\textsc{GnnPruner}\xspace}

\newcommand{\random}{\textsc{Random}\xspace}
\newcommand{\topk}{\textsc{Top-k}\xspace}

\usepackage{algorithm}
\usepackage{algorithmic}
\usepackage{amsmath}

\usepackage{newfloat}
\usepackage{listings}
\usepackage{multirow}
\usepackage{subfigure}
\usepackage{booktabs}
\usepackage{graphicx}
\usepackage{xcolor}
\usepackage{url}

\usepackage[accepted]{aistats2025}
\usepackage[numbers,compress]{natbib}

\begin{document}

%

%

\twocolumn[

\aistatstitle{Theoretically Grounded Pruning of Large Ground Sets for Constrained, Discrete Optimization}

\aistatsauthor{Ankur Nath \And Alan Kuhnle }

\aistatsaddress{ Department of Computer Science and Engineering, Texas A\&M University } ]

\begin{abstract}
  Modern instances of combinatorial optimization problems
  often exhibit billion-scale ground sets, which have many
  uninformative or redundant elements. In this work, we
  develop light-weight pruning algorithms to quickly
  discard elements that are unlikely to be part of an
  optimal solution. Under mild assumptions on the instance,
  we prove theoretical guarantees on the fraction of the
  optimal value retained and the size of the resulting
  pruned ground set. 
  Through extensive experiments on real-world datasets for various applications, we demonstrate that our algorithm, \qs{},
  efficiently prunes over 90\% of the ground set
  and outperforms state-of-the-art classical and
  machine learning heuristics for pruning. %

\end{abstract}


\section{INTRODUCTION}
In many data science and machine learning
tasks, the optimization of a discrete function
is required. For example, \textit{subset selection}
problems \citep{Wei2015,Kim2020}, 
such as selecting the most informative features from a dataset for model training.
As another example, consider the problem of selecting a subset of items to display
in a recommendation system \citep{Mehrotra2023,Ko2022}, where the goal is to maximize user engagement.

In this work, we consider maximization of an
objective function $f$ defined on subsets
of $\uni$, subject to a
knapsack constraint: a modular
cost function $c$ on the elements
is restricted to be at most the budget $\kappa$.
Denote by
$f_{\kappa}(X) = \max_{S \subseteq X : c(S)\le \kappa } f(S),$
where $X \subseteq \uni$. 
For many applications of this problem, the ground
set $\uni$ is massive, with size $n$ in the billion-scale
or larger. On the other hand, the budget constraints
are frequently such that an optimal set has only a few elements.
For example,
in viral marketing campaigns \citep{Kempe2003a},
a company may have a budget to promote
only a limited number of products while still aiming to
maximize overall sales or customer reach.
Intuitively, in these cases, the vast majority of elements of $\uni$ are irrelevant
to finding $f_{\kappa}(\uni)$.

Therefore, rather than solving $f_{\kappa}(\uni)$ directly,
it may be beneficial to produce a small set $\uni' \subseteq \uni$
of relatively promising candidate elements. Once $\uni'$
is obtained, an expensive
heuristic or exact algorithm may be employed to produce a feasible
solution.
Further, since elements not in $\uni'$ are discarded,
we desire a range of budgets $[\kappa_{\min}, \kappa_{\max}]$
to be supported, so that the pruned
ground set has value beyond a single use.
Formally, we have the following problem definition.

\textbf{Problem definition (Pruning).}
Given an objective function $f: 2^\uni \to \mathbb{R}_+$,
modular cost function $c: \uni \to \mathbb{R}_+$,
and budget range $[\kappa_{\min}, \kappa_{\max}]$,
produce $\uni' \subseteq \uni$,
such that
\begin{itemize}
\item $|\uni'| = \oh{ F( \kappa_{\max}, \kappa_{\min}, c ) \polylog(n)},$ where $F$ is a function and $n = |\uni |$; 
\item there exists $\alpha \in [0,1]$, such that, for any $\tau \in [\kappa_{\min}, \kappa_{\max}]$,
  it holds that $f_{\tau}( \uni' ) \ge \alpha f_{\tau }( \uni ).$
\end{itemize}

Existing methods for the pruning task \citep{Zhou2017,Manchanda2020a,Ireland2022,Tian2024} are heuristics
that 1) are only formulated to solve one instance of size-constrained maximization;
2) provide no guarantee on the size of $\uni'$; and
3) provide no guarantee on the fraction of the optimal value retained in $\uni'$ after the pruning process. Thus, in this work, we are motivated by the following
questions:
\begin{center}
  \textit{Is it possible to develop pruning algorithms with theoretical guarantees that are useful beyond solving one problem instance?
    If so, what assumptions are required on the objective function and the problem instance?
    Are the resulting algorithms practical and competitive with existing heuristics?}
\end{center}
\begin{figure*}[t] 
  \centering
  \subfigure[] {
    \includegraphics[width=0.45\textwidth]{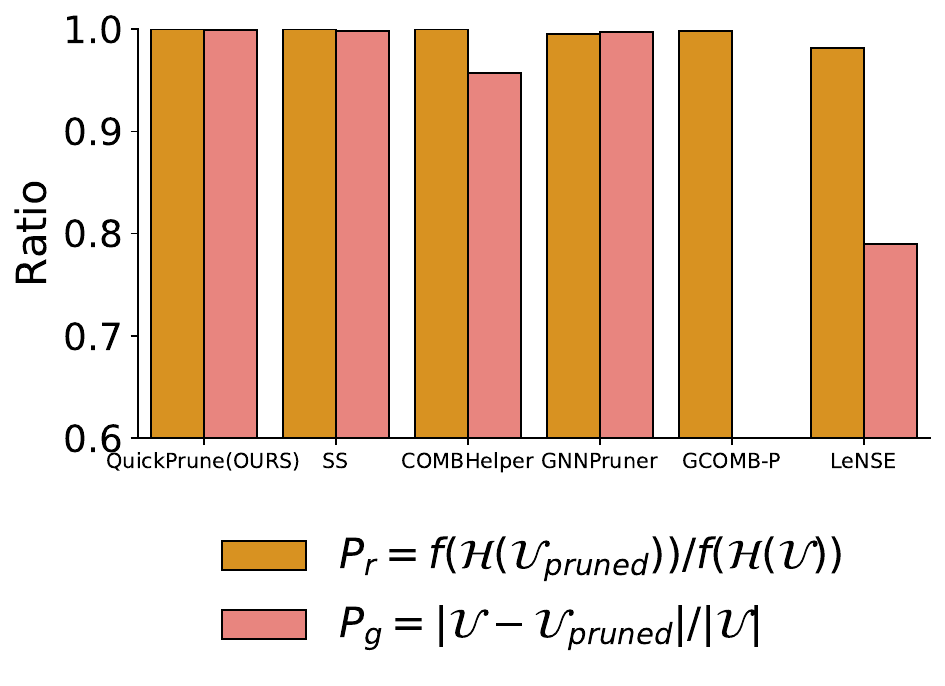}
  }
  \subfigure[] {
    \includegraphics[width=0.45\textwidth]{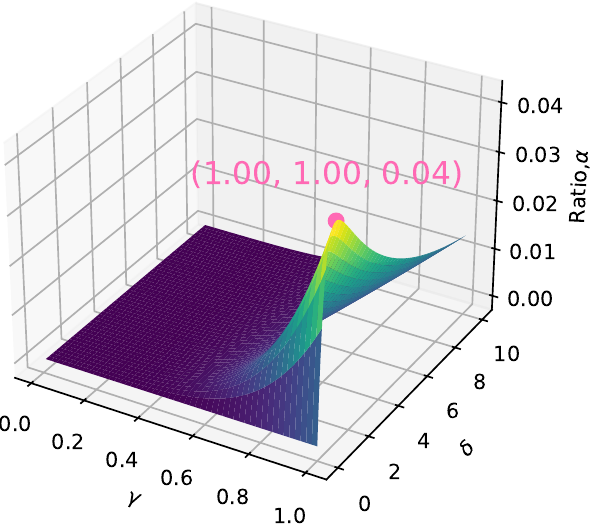}
  }
  \caption{\textbf{(a):} Typical empirical results of \qs versus competing methods on an instance of the MaximumCover problem: \qs retains
    99.99\% of the optimal value while pruning $99.95\%$ of the ground set.
  \textbf{(b):} Plot of the pruning ratio $\alpha (\epsi, \delta, \gamma )$ of Theorem \ref{thm:multi}, as a function of the parameter $\delta$ of \qs and the $\gamma$-submodularity of the objective function $f$. Here, $\epsi$ is fixed to $0.01$.} 
  \label{fig:intro-fig}
\end{figure*}

\textbf{Contributions.} 
 We introduce \qs{}, a lightweight pruning method that
 prunes the original ground set $\uni$
 for a range of budgets in a single pass
 through the ground set.
 We emphasize that lightweight pruning methods
 are needed since an expensive pruning process defeats the
 purpose of pruning. In this regard, our algorithm
 processes elements one-by-one -- once an element is rejected,
 it is never re-considered and may be safely discarded.
 \qs{}  makes at most $\oh{\log ( \kappa_{\max} / \kappa_{\min} )}$
 queries to $f$ per element processed. 
 Moreover, we prove theoretical bounds
 on the quality of the pruned ground set $\uni'$ of \qs{}, and also
 on its size.

 Specifically, given budget range $\kappa_{\text{min}} \le \kappa_{\text{max}}$, we show that
 the output of \qs{} satisfies
$| \uni' | = \oh{ \log \left( \frac{ \kappa_{\max} }{ \kappa_{\min} } \right) \cdot \left( \frac{\kappa_{\max}}{ c_{\min}} \right) \log ( n)}$, where $c_{\min} = \min_{u \in \uni} c(u)$. 
Further, if the objective function $f$ is $\gamma$-weakly
submodular\footnote{Submodularity is a diminishing-returns condition ubiquitous to many applications, discussed further in Section \ref{sec:prelim}.}, and if a mild assumption\footnote{Namely, that no single element uses almost all of the budget. See Section \ref{sec:alg} for the precise formulation.} on the costs of elements in
 an optimal solution of $f_{\tau}(\uni )$ holds, we show
 that
 $f_{\tau} ( \uni' ) \ge \alpha( \epsi, \delta, \gamma ) \cdot f_{\tau} ( \uni )$,
 for any $\tau \in [\kappa_{\min}, \kappa_{\max} ]$, where $\epsi, \delta$
 are parameters of the algorithm. A plot of $\alpha$ is shown in Fig. \ref{fig:intro-fig}(b);
 If $\gamma = 1$ and $\delta = 1$,
 $f_{\tau} ( \uni' ) \ge \left( \frac{1}{24} - \epsilon \right) f_{\tau} ( \uni )$.
 To the best of our knowledge, all previous methods for this
 problem are heuristics with no performance guarantee, either on
 the size of the pruned universe or on the fraction of the optimal
 value retained after pruning. We give a comprehensive discussion
 of related work in Section \ref{sec:related-work}. 

 
 To show the practical effectiveness of our algorithm,
 we evaluate it against several
 state-of-the-art classical and machine learning heuristics
 for the pruning problem across four different optimization contexts.
 The methods are evaluated by two metrics, the fraction $P_g$ of the
 original ground set that is pruned (higher is better); and the fraction $P_r$ of the
 value of a solution by a given algorithm $\mathcal H$ for the problem
 that is retained (higher is better);
 \ie $P_r = f( \mathcal H ( \uni' ) ) / f( \mathcal H ( \uni ) )$.
 Empirically, as shown in Fig. \ref{fig:intro-fig},
 our algorithm outperforms competing methods on both metrics,
 and achieves substantial reductions in
 ground set size (typically over 90\%) while nearly preserving the
 value of $f_\kappa $ across multiple budgets.

 \textbf{Organization.} The rest of this paper
 is organized as follows. In Section \ref{sec:related-work},
 we describe the relationship of our contribution to existing
 work; in Section \ref{sec:prelim}, we introduce preliminaries
 and notation. In Section \ref{sec:alg}, we describe our
 \qs algorithm and prove Theorem \ref{thm:multi}, which
 summarizes its theoretical properties. In Section \ref{sec:exp},
 we conduct our empirical evaluation and comparison to existing
 heuristics for the pruning problem. In Section \ref{sec:con},
 we conclude the paper and discuss future directions.
 In Appendices, we describe omitted details and proofs from the main text.

\subsection{Related Work} \label{sec:related-work}
\textbf{(Weakly) Submodular Optimization.}
\textit{Submodularity} is a notion of diminishing
returns that is satisfied or partially satisfied
by many, varied objective functions \citep{Demeniconi2021, Feige, Feige1996, Feige2011, Feige2011a, Feige2013, Gharan2010, Gharan2011, Gharan2011a, Gharan2011b, Gupta2010a, Khot2001, Kothawade2022, Lee2009, Nemhauser1978b, Uziahu2023, Vondrak2013}.
Constrained non-submodular maximization has proven useful for
data summarization (\eg \citep{Demeniconi2021,Kothawade2022}), feature selection \citep{Elenberg2018},
reduction of training set size 
for deep learning methods \citep{Killamsetty2023}.
Submodularity occupies a role for discrete optimization
analagous to the role of convexity in continuous optimization.
The typical model (\eg \citep{Feige, Gharan2011a, Gupta2010a}) is that the function $f$ is
available to an algorithm as a black-box orcale
that returns $f(S)$ when queried with set $S$. However,
these function evaluations are very expensive, so we
desire to minimize the \textit{query complexity} of an algorithm. 


\textbf{Coresets.}
Our notion of pruning is related to the
idea of a coreset \citep{Braverman2022, Feldman2020b, Indyk2014, Kogan2015, Liu2019, Mirrokni2015, Mirzasoleiman2020, Tukan2020, Tukan2022, Yang2023a, Zhang2022e}; in which
a small set is to be selected to
best summarize a set of
points with respect to a desired
set of queries; for example,
construct a weighted graph that
approximately preserves the values
of every cut in an original graph
\citep{Kogan2015}. 
The concept of coreset does not
have a precise definition,
and thus our 
terminology
\textit{pruned universe}
could be
recast as coreset. 
Coresets have been applied
to \eg graph summarization \citep{Liu2019},
data-efficient training of ML models \citep{Mirzasoleiman2020,Yang2023a},
pruning neurons from neural networks \citep{Tukan2022}.
For the types of optimization
problem we consider, \ie
constrained maximization with objective
function available as a value query oracle,
there have been only a small
number of attempts to produce a coreset
\citep{Indyk2014,Mirrokni2015}.
The closest to our setting
is the method of \citet{Mirrokni2015}.
Their notion of \textit{randomized
  composable coreset}
is used for distributed computation
and is only applicable to
a single size constraint. 
Also, their notion of coreset
is much stronger than what we require
in this work. 

\textbf{Pruning via supervised learning methods.}
There have been multiple ML-based heuristics
proposed to accomplish the pruning task
for a single instance of a combinatorial optimization
problem.
The first line of work \citep{Lauri2019, Lauri2023, Sun2021b, Sun2021c, Zhang2022b}
seeks to train a binary classification algorithm to predict
whether a data element belongs to an optimal solution.
These approaches
consider graph-based optimization problems and
train a linear classifer on vertices
by hand-crafting a set of local features for each vertex,
such as centrality measures,
solutions to linear programming relaxations, and so on; a classifier, such as random forests, is then trained on small instances that can be exactly solved. The goal is to correctly predict when a vertex does \textit{not} belong to an optimal solution,
and hence can be safely pruned. 
These methods rely on hand-crafted features
customized to each problem; which themselves may
be superlinear to compute, such as a
solution to an LP relaxation of the problem
or the eigenvector centrality on the original (unpruned)
graph.
A recent generalization (COMBHelper) of this approach uses a GCN (graph convolution network)
in place of a linear classifer \citep{Tian2024}.
We compare to COMBHelper empirically in Section \ref{sec:exp}. 

\textbf{Other heuristics for pruning.}
GCOMB \citep{Manchanda2020a}
uses a weighted degree heuristic
with a probabilistic greedy algorithm,
combined with a supervised learning approach, to produce a pruned
ground set. On the other hand, LeNSE \citep{Ireland2022} trains an RL agent to navigate from an initial subgraph to a pruned subgraph that most likely contains the optimal solution for a specific optimization problem via a local search procedure on subgraphs.
Another heuristic is that of \citet{Zhou2017}, which uses ideas from submodularity
to extract a promising pruned universe. We empirically compare our algorithm
to each of these methods in Section \ref{sec:exp}. In contrast to our work,
all of these are formulated
for a single instance of size-constrained maximization and have no theoretical
guarantees of any kind, \textit{including that a nonempty feasible solution remains after the
pruning process}. 


\subsection{Preliminaries and Notation} \label{sec:prelim}
For natural number $n$, we use the notation $[n] = \{0, 1, \ldots, n-1 \}$,
and for a function $f$ understood from context: $\marge{T}{S} = f( S \cup T ) - f(S)$
is the \textit{marginal gain} of set $T$ to set $S$. 
A non-negative set function on ground set $\uni$, $f : 2^\uni \to \mathbb R^+$,
is \textit{submodular} iff
$\forall S \subseteq \uni, \forall T \subseteq \uni, f( S \cap T ) + f( S \cup T ) \le f(S) + f(T).$
Roughly speaking, this means the whole is not greater than the sum of its parts.
An equivalent characterization is the following notion of diminishing-returns:
$ \forall S \subseteq T \subseteq \uni, \forall x \not \in T, \marge{x}{T} \le \marge{x}{S}.$
A non-negative, set function $f$ 
is \textit{monotone} iff
$\forall S \subseteq T \subseteq \uni, f( S ) \le f( T ).$
A submodular function is not necessarily monotone.

In this work, we consider the following relaxation of
submodularity: the monotone function $f$ is $\gamma$-submodular
iff $\gamma$ is the maximum
value in $[0,1]$ such that $\forall S \subseteq T \subseteq \uni, \forall x \not \in T, \gamma \marge{x}{T} \le \marge{x}{S}.$
Observe that $f$ is submodular iff $\gamma = 1$. 
Occasionally in the literature, $\gamma$ is termed the \textit{diminishing-returns ratio} of a function
\cite{Bian2017b,Kuhnle2018}.
We remark that monotonicity is needed for the definition of $\gamma$-submodular,
and throughout our analysis, we assume that the objective function $f$ is monotone. 


\section{QUICKPRUNE: A PRUNING ALGORITHM WITH GUARANTEES} \label{sec:alg} 
In this section, we describe the pruning algorithms introduced in this paper.
In Section \ref{sec:prune-single},
we describe the pruning algorithm for a single knapsack constraint.
The main pruning algorithm (\qs, Alg. \ref{alg:multi}) runs multiple
copies of the single constraint algorithm in parallel, as detailed
in Section \ref{sec:prune-multi}.

To prove theoretical guarantees for \qs, we make
the following assumption on the costs. Intuitively, the assumption
says that no element of an optimal solution consumes a very large
fraction of the total budget. 
\begin{assumption}[No Huge Items (\nhi)]
  Given an instance $(f,c,\kappa)$ of \sm,
  the instance satisfies the assumption with $\eta > 0$
  if there exists an optimal solution $O$ to the
  instance, such that, for all $o \in O$, $c(o) \le \kappa (1-\eta)$.
\end{assumption}
For example, in the case of size constraint, this assumption is
satisfied if $\kappa \ge 2$ and $\eta \le 1/2$. 

Our theoretical guarantees for the main pruning algorithm are
summarized in Theorem \ref{thm:multi}, which is proved in Section \ref{sec:prune-multi}.
Observe that we show a constant factor $C$ of the optimal value is retained
for all budgets in the range $[\kappa_{\min}, \kappa_{\max}]$; the desired
range of budgets to support by the user. The constant
depends mainly on $\gamma$, the submodularity parameter of $f$,
and the input parameter $\delta$; $C$ is optimized
with $\delta = 1$ at the value $1/24 - \epsi$.
Also, we show a size bound on the size of $\uni'$, the pruned ground set,
which depends logarithmically on the size $n$ of $\uni$, as well as the values
of the maximum and minimum budgets and the minimum cost of an element. 
\begin{theorem} \label{thm:multi}
  Let $\kappa_{\min} < \kappa_{\max}$, let $0 < \eta \le 1/2$, and let $f, c$ be
  $\gamma$-submodular and modular functions, respectively.
  Suppose that for all $\kappa' \in [\kappa_{\min}, \kappa_{\max}]$,
  the instance $(f,c,\kappa')$ satisfies the \nhi assumption with $\eta$.
  Let $\uni'$ be the output of \qs (Alg. \ref{alg:multi})
  with parameters
  $(f,c,\kappa_{\min},\kappa_{\max}, \delta, \epsi, \eta)$.
  Then, for all $\kappa' \in [\kappa_{\min}, \kappa_{\max}]$,
  it holds that
  $f_{\kappa'}( \uni' ) \ge
  \frac{\delta \gamma^5(1 - \epsi \gamma^{-1})}{6(\delta \gamma^2 + 1)( 1 + \gamma^{-1} \delta ) }  f_{\kappa'}( \uni ).$
  And, $| \uni' | = \oh{ \log( \frac{ \kappa_{\max} }{ \kappa_{\min} } ) \left( 1 + \frac{\kappa_{\max} }{\delta c_{min} } \right) \log ( n / \epsi ) }.$
\end{theorem}
\subsection{Pruning for a Single Knapsack Constraint} \label{sec:prune-single}
First, we target the case of a single constraint.
Given a single budget value $\kappa$, we formulate an algorithm to prune for the
instance $(f, \uni, \kappa,c)$. That is, we produce an algorithm that, for some
$\alpha > 0$, produces $\uni'$ such that $f_{\kappa}( \uni' ) \ge \alpha  f_{\kappa }( \uni )$
and $|\uni'| \le \oh{F( \kappa, c) \polylog(n)}. $

The algorithm \qss is given in Alg. \ref{alg:single}. In addition to the problem
instance, it takes parameter $\delta > 0$, which impacts the size of the resulting
pruned set $\uni'$ by adjusting the condition to add elements to $\uni'$;
and parameter $\epsi > 0$, which controls how aggressively the algorithm deletes elements.
In overview, the algorithm works by taking one pass through the ground set $\uni$ and
processing elements one-by-one. An element is added to $\uni'$ if it meets the condition
on Line \ref{line:add}: $\marge{e}{A} \ge \frac{ \delta c(e) f(A) }{ \kappa }$; intuitively,
this means the element should increase the value of $A$ (the pruned set)
an amount proportional to $c(e) / \kappa$ to be worth retaining.

To ensure the size of $\uni'$ doesn't grow too large, a deletion condition is checked
on Line \ref{line:delete}, which intuitively asks if the value of $A$ has increased by
a large factor from the previous checkpoint; if so, the original elements of $A$ are
discarded. Submodularity is used to bound the amount of value lost from the deletion,
and the condition on element addition bounds the maximal number of elements required for
the increased value.

The theoretical properties are summarized in the following theorem, proven in Section \ref{sec:single-proofs}.
\begin{theorem} \label{thm:single}
  Let Alg. \ref{alg:single} be run on instance $(\uni, c, f, \kappa)$, with parameters $\epsi, \delta > 0,$
  such that $f$ is $\gamma$-submodular. 
  Then, Alg. \ref{alg:single} produces pruned universe $\uni'$, such that
  $| \uni ' | < 2 \left( 1 + \frac{\kappa }{\delta c_{min} } \right) \log ( n / \epsi ) + 3,$
  and there exists $A' \subseteq \uni'$, such that $c(A') \le \kappa$ and
  $f(A') \ge \frac{\delta \gamma^4(1 - \epsi \gamma^{-1})}{2(\delta \gamma^2 + 1)( 1 + \gamma^{-1} \delta ) } \opt$.
\end{theorem}
\begin{algorithm}[t]
	\caption{\qss:  Pruning for a single constraint.}
	\begin{algorithmic}[1] \label{alg:single}
		\STATE \textbf{Input:} Instance $(f, \uni, c, \kappa)$, size-control parameter $\delta > 0$, deletion parameter $\epsi > 0$
		\STATE \textbf{Output:} Pruned ground set, $\uni'$
		\STATE \textbf{Initialize:} $A \gets \emptyset$, $a^* \gets \emptyset$, $A_s \gets \emptyset$
                \FOR {$e \in \uni$}
                \IF{ $c(e) > \kappa$  }
                \STATE \textbf{continue}
                \ENDIF{}
                \IF{$ \marge{e}{A} \ge \frac{\delta c(e) f(A) }{ \kappa }$ \label{line:add}}
                \STATE $A \gets A + e$
                \ENDIF{}
                \IF{ $f(e) > f(a^*)$  }
                \STATE $a^* \gets e$
                \ENDIF
                \IF{ $f(A) > \frac{ n }{ \epsi } f(A_s) $\label{line:delete} }
                \STATE $A \gets A \setminus A_s$ \label{line:delete-set}
                \STATE $A_s \gets A$
                \ENDIF
                \ENDFOR
                \STATE \textbf{return} $\uni' \gets A + a^*$ 
	\end{algorithmic}
\end{algorithm}
\subsubsection{Proof of Theorem \ref{thm:single}} \label{sec:single-proofs}
Due to space constraints, omitted proofs are provided in
Appendix \ref{apx:proof}. We require the following fact about geometrically increasing sequences of real
numbers. For our purposes, the sequence $(y_i)$ will assume the values of $f(A)$ as elements are
added. 
\begin{fact} \label{fact:1}
  Let $(y_i)_{i=1}^m$ be a sequence of positive real numbers, such that, for some $\beta > 0$, it holds that $y_i \ge (1 + \beta)y_{i-1}$, for all $i \in [m]$. Let $\gamma > 0$. Then, if $m \ge \frac{\beta + 1}{\beta}\log \gamma^{-1} $, it holds that $y_m \ge y_1 / \gamma$.
\end{fact}
Using Fact \ref{fact:1}, we bound the number of elements in $\uni'$. Intuitively, the geometric increase
in the value of $f(A)$ yields a bound on the maximum number of elements until the deletion condition
on Line \ref{line:delete} is triggered. 
\begin{proposition} \label{prop:uprime-size}
  The size of $\uni'$, as returned by Alg. \ref{alg:single}, satisfies
  $|\uni'| < 2 \left( 1 + \frac{\kappa }{\delta c_{min} } \right) \log ( n / \epsi ) + 3,$
  where $c_{min} = \min_{a \in \uni } c(a).$
\end{proposition}
Next, we bound how much value may have been lost from deletion throughout the execution
of the algorithm. In the following proposition, $\hat A$ is all elements ever added to
$A$, and $\dot A$ is the actual set obtained by the algorithm after all deletions. We
show at most an $\epsi\gamma^{-1}$-fraction of value is lost. 
\begin{proposition} \label{prop:ahat-adot}
  Suppose $f$ is $\gamma$-submodular.
  Let $A_i$ be the value of $A$ after the execution of iteration $i$ of the \textbf{for} loop,
  for $i = 1$ to $n$, and $A_0 = \emptyset$. 
  Let $\hat A = \bigcup A_j, \dot A = A_n$. Then $f( \dot A ) \ge (1 - \gamma^{-1} \epsi )f( \hat A)$. 
\end{proposition}
So far, we have established a size bound on $\uni'$ and bounded the value of
$\uni'$ lost from deletion. Next, we turn to showing that there exists a feasible set
inside $\uni'$ that has a constant fraction of the optimal value.
We start by showing in Lemma \ref{lemma:value-inside-A}
and Proposition \ref{prop:aprime} that
a set $A^*$ exists within $\uni'$ that has a constant fraction
of $f(\uni')$. 
\begin{lemma} \label{lemma:value-inside-A}
  Let $f$ be $\gamma$-submodular.
  Let $\delta, \kappa > 0$. Suppose $A_i = \{ a_1, \ldots a_i \}$ is a sequence of sets satisfying
  (1) $c(a_i) \le \kappa$, and (2) $\marge{ a_i }{ A_{i - 1} } \ge \gamma \frac{\delta c( a_i ) f( A_{i -1 } ) }{ \kappa }$,
  for each $i$ from $1$ to $m$. Let $f( a^* ) \ge \max_{i \in [m]} f(a_i)$, and $c(a^*) \le \kappa$.
  Then there exists $A^* \subseteq A_m + a^*$, such that $c( A^* ) \le \kappa$, and
  $f( A^* ) \ge \frac{ \delta \gamma^4}{2(\delta\gamma^2 + 1)} f(A_m)$.
\end{lemma}
\begin{proposition} \label{prop:aprime}
  Let $\uni' = \dot A + a^*$ have its value at termination of \qss.
  There exists $A' \subseteq \uni' = \dot A + a^*$, such that  $f(A') \ge \frac{\delta \gamma^4}{2(\delta \gamma^2 + 1)} f(\dot A )$
  and $c(A') \le \kappa$. 
\end{proposition}
Next, we relate the value of $f(\uni')$ to
$\opt = f_\kappa ( \uni )$ in Proposition \ref{prop:opt-ahat}.
\begin{proposition} \label{prop:opt-ahat}
  Suppose Alg. \ref{alg:single} is run on instance $(\uni,c,f,\kappa )$ with
  parameters $\delta, \epsi > 0$. 
  Let $\hat A = \bigcup_{j \in [n]} A_j$ be all elements added to $A$.
  Then $f( \hat A ) \ge \opt / (1 + \gamma^{-1}\delta )$, where $\opt = \max_{S \subseteq \uni: c(S) \le \kappa } f(S).$
\end{proposition}
\begin{proof}[Proof of Theorem \ref{thm:single}]
  Finally, we are ready to put all the pieces together.
  Assume the hypotheses of the theorem statement. 
  The size bound on $\uni'$ follows by Proposition \ref{prop:uprime-size}.
  Let $A_j$ be the value of the set
  $A$ at the beginning of the $j$th iteration of \qss.
  Let $\dot A = A_{n+1}$ be the final value of $A$,
  $\hat A = \bigcup_{i \in [n + 1]} A_i$.
  Let $A'$ be as guaranteed by Proposition \ref{prop:aprime}.
  Then$ f(A') \overset{(a)}{\ge} \frac{\delta \gamma^4 }{2(\delta \gamma^2 + 1)} f( \dot A ) 
  \overset{(b)}{\ge} \frac{\delta \gamma^4(1 - \epsi \gamma^{-1})}{2(\delta \gamma^2 + 1)} f( \hat A )$,
  where (a) is by Prop. \ref{prop:aprime},
  (b) is by Prop. \ref{prop:ahat-adot},
  and the result follows from Prop. \ref{prop:opt-ahat}.
\end{proof}

\subsection{The Pruning Algorithm: \qs} \label{sec:prune-multi} 
\begin{algorithm}[t]
	\caption{\qs: The Pruning Algorithm.}
	\begin{algorithmic}[1] \label{alg:multi}
          \STATE \textbf{Input:} Instances $(f, \uni, c, [\kappa_{\min}, \kappa_{\max}])$,  parameters $\delta > 0$, $\epsi > 0$,
          $0 < \eta \le 1/2$.
		\STATE \textbf{Output:} Pruned ground set, $\uni'$
		\STATE $\mathcal B = \{ \tau_i = \kappa_{\max}(1 - \eta )^i : i \in \mathbb Z, (1 - \eta) \kappa_{\min} \le \tau_i \le \kappa_{\max}\}$\label{line:multi-budgets}
                \STATE Initialize a copy of $\mathcal QP_{\tau}$ of \qss for each $\tau \in \mathcal B$, with parameters $\epsi, \delta$.
                \FOR {$e \in \uni$}
                \STATE Pass $e$ to $\mathcal QP_{\tau}$ for all $\tau \in \mathcal B$.
                \ENDFOR
                \STATE \textbf{return} the union of all sets returned by all instances $\mathcal QP_{\tau}$
	\end{algorithmic}
\end{algorithm}
In this section, we describe the main pruning algorithm \qs (Alg. \ref{alg:multi})
and prove Theorem \ref{thm:multi}. 

The algorithm \qs works as follows. As input, it takes instances with a range of
budgets $[\kappa_{\min}, \kappa_{\max}]$, the parameters $\epsi, \delta$ for
\qss, and a parameter $\eta > 0,$ such that all of the instances satisfy
Assumption \nhi with $\eta$. The algorithm then runs $\log ( \kappa_{\max} / \kappa_{\min} )$
copies of \qss in parallel -- one each for budget $\tau_i = \kappa_{max}(1 - \eta )^i$ in the
range $[( 1 - \eta ) \kappa_{\min}, \kappa_{\max} ]$. It returns the union of all pruned sets
returned from each copy of \qss.

To establish Theorem \ref{thm:multi}, we first show the following proposition,
with implications to how Assumption \nhi can relate an optimal solution of a given
budget to one of the instances handled by one of the copies of \qss. 
\begin{proposition} \label{prop:part}
  Suppose instance $(f,c,\kappa)$ satisfies the \nhi assumption
  with $0 < \eta \le 1/2$. 
  Then, there exists an optimal solution $O$ to the instance, such
  that $O$ can be partitioned into at most three sets $\{ O_i : i \in [3] \}$, such
  that $c(O_i) \le \kappa (1 - \eta )$, for each $i \in [3]$. 
\end{proposition}
\begin{proof}[Proof of Theorem \ref{thm:multi}]
  Assume the hypotheses of the theorem, and let
  $\kappa' \in [\kappa_{\min}, \kappa_{\max}]$. By the choice of
  the set $\mathcal B$ on Line \ref{line:multi-budgets}, there
  exists a $\tau \in \mathcal B$, such that $\kappa'(1 - \eta) \le \tau \le \kappa'$.
  Let $O \subseteq \uni$ be an optimal solution to $f_{\kappa'}( \uni )$
  satisfying Assumption \nhi. By Proposition \ref{prop:part}, $O$ can be partitioned
  into at most three sets $\{ O_i \}$, such that $c( O_i ) \le \kappa ( 1 - \eta ) \le \tau$.

  Moreover, by Theorem \ref{thm:single}, there exists $X \subseteq \uni'$, such that
  1) $c(X) \le \tau$, and 
  2) $f(X) \ge
  \alpha f( O_i )$,
  with $\alpha =  \frac{\delta \gamma^4(1 - \epsi \gamma^{-1})}{2(\delta \gamma^2 + 1)( 1 + \gamma^{-1} \delta ) }$.
  Therefore, $3f(X) \ge \alpha \sum_{i = 1}^3 f( O_i ) \ge \alpha \gamma f(O)$, where the last inequality follows from submodularity.
  Thus, $f(X) \ge \alpha \gamma f(O) / 3$.
  Finally, the size bound on $\uni'$
  follows from Prop. \ref{prop:uprime-size} and the size of $\mathcal B$. 
\end{proof}

\section{EMPIRICAL EVALUATION} \label{sec:exp}

\begin{table*}[t]
\centering
\caption{Comparison of pruning algorithms for size constraint experiments (best combined metric in bold); *Values as reported in  \cite{ireland2022lense} and ``--" denotes no reasonable result is achieved by the corresponding algorithm under the time constraint.}
\vspace{1em}
\label{tab:size_constraint}
\resizebox{\textwidth}{!}{%
\begin{tabular}{|ccccccccccccccccccc|}
\hline
\multicolumn{1}{|c|}{}         & \multicolumn{3}{c|}{\qs}                                               & \multicolumn{3}{c|}{SS}                                               & \multicolumn{3}{c|}{\gcomb*}                                         & \multicolumn{3}{c|}{\comb}                                            & \multicolumn{3}{c|}{\lense*}                                          & \multicolumn{3}{c|}{\gnn}                        \\ \hline
\multicolumn{1}{|c|}{Graph}    & $P_r \uparrow$  & $P_g$ $\uparrow$ & \multicolumn{1}{c|}{$C \uparrow$} & $P_r \uparrow$  & $P_g$ $\uparrow$ & \multicolumn{1}{c|}{$C\uparrow$} & $P_r \uparrow$ & $P_g$ $\uparrow$ & \multicolumn{1}{c|}{$C\uparrow$} & $P_r \uparrow$  & $P_g$ $\uparrow$ & \multicolumn{1}{c|}{$C\uparrow$} & $P_r \uparrow$ & $P_g$ $\uparrow$ & \multicolumn{1}{c|}{$C\uparrow$} & $P_r \uparrow$  & $P_g$ $\uparrow$ & $C\uparrow$ \\ \hline
& \multicolumn{18}{c|}{\textbf{Maximum Cover}}  \\ \hline
\multicolumn{1}{|c||}{Facebook}& 1.0000& 0.9953& \multicolumn{1}{c||}
{\textbf{0.9953}}& 0.9884& 0.9027& \multicolumn{1}{c||}
{0.8922}& 0.9270& 0.0700& \multicolumn{1}{c||}
{0.0649}& 1.0000& 0.3840& \multicolumn{1}{c||}
{0.3840}& 0.9660& 0.0700& \multicolumn{1}{c||}
{0.0676}& 1.0000& 0.7697& \multicolumn{1}{c|}
{0.7697}\\
\multicolumn{1}{|c||}{Wiki}& 0.9998& 0.9422& \multicolumn{1}{c||}
{\textbf{0.9420}}& 0.9559& 0.9346& \multicolumn{1}{c||}
{0.8934}& 0.9900& 0.0300& \multicolumn{1}{c||}
{0.0297}& 1.0000& 0.4528& \multicolumn{1}{c||}
{0.4528}& 1.0940& 0.3400& \multicolumn{1}{c||}
{0.3720}& 1.0000& 0.9199& \multicolumn{1}{c|}
{0.9199}\\
\multicolumn{1}{|c||}{Deezer}& 0.9606& 0.9797& \multicolumn{1}{c||}
{0.9411}& 0.9870& 0.9855& \multicolumn{1}{c||}
{\textbf{0.9727}}& 0.9940& 0.1300& \multicolumn{1}{c||}
{0.1292}& 1.0000& 0.7278& \multicolumn{1}{c||}
{0.7278}& 0.9790& 0.7500& \multicolumn{1}{c||}
{0.7343}& 1.0000& 0.9151& \multicolumn{1}{c|}
{0.9151}\\
\multicolumn{1}{|c||}{Slashdot}& 1.0000& 0.9925& \multicolumn{1}{c||}
{\textbf{0.9925}}& 0.9824& 0.9889& \multicolumn{1}{c||}
{0.9715}& 1.0000& 0.0200& \multicolumn{1}{c||}
{0.0200}& 1.0000& 0.9844& \multicolumn{1}{c||}
{0.9844}& 0.9790& 0.6900& \multicolumn{1}{c||}
{0.6755}& 1.0000& 0.9810& \multicolumn{1}{c|}
{0.9810}\\
\multicolumn{1}{|c||}{Twitter}& 0.9929& 0.9911& \multicolumn{1}{c||}
{\textbf{0.9841}}& 0.9306& 0.9893& \multicolumn{1}{c||}
{0.9206}& 0.9970& 0.1700& \multicolumn{1}{c||}
{0.1695}& 1.0000& 0.5654& \multicolumn{1}{c||}
{0.5654}& 0.9890& 0.3300& \multicolumn{1}{c||}
{0.3264}& 0.9987& 0.9793& \multicolumn{1}{c|}
{0.9780}\\
\multicolumn{1}{|c||}{DBLP}& 0.9951& 0.9957& \multicolumn{1}{c||}
{\textbf{0.9908}}& 0.9945& 0.9963& \multicolumn{1}{c||}
{\textbf{0.9908}}& 0.9990& 0.0300& \multicolumn{1}{c||}
{0.0300}& 1.0000& 0.1818& \multicolumn{1}{c||}
{0.1818}& 0.9900& 0.9000& \multicolumn{1}{c||}
{0.8910}& 1.0000& 0.8705& \multicolumn{1}{c|}
{0.8705}\\
\multicolumn{1}{|c||}{YouTube}& 1.0000& 0.9995& \multicolumn{1}{c||}
{\textbf{0.9995}}& 0.9999& 0.9987& \multicolumn{1}{c||}
{0.9986}& 0.9980& 0.0700& \multicolumn{1}{c||}
{0.0699}& 1.0000& 0.9572& \multicolumn{1}{c||}
{0.9572}& 0.9820& 0.7900& \multicolumn{1}{c||}
{0.7758}& 0.9947& 0.9967& \multicolumn{1}{c|}
{0.9914}\\
\multicolumn{1}{|c||}{Skitter}& 0.9985& 0.9997& \multicolumn{1}{c||}
{0.9982}& 0.9857& 0.9991& \multicolumn{1}{c||}
{0.9848}& 0.9990& 0.1000& \multicolumn{1}{c||}
{0.0999}& --& --& \multicolumn{1}{c||}
{--}& 0.9760& 0.7000& \multicolumn{1}{c||}
{0.6832}& 0.9993& 0.9891& \multicolumn{1}{c|}
{\textbf{0.9884}}\\
\hline
& \multicolumn{18}{c|}{\textbf{Maximum Cut}}  \\ \hline
\multicolumn{1}{|c||}{Facebook}& 0.9886& 0.7935& \multicolumn{1}{c||}
{0.7845}& 0.9928& 0.9027& \multicolumn{1}{c||}
{\textbf{0.8962}}& 0.8130& 0.9500& \multicolumn{1}{c||}
{0.7723}& 1.0000& 0.1538& \multicolumn{1}{c||}
{0.1538}& 1.0000& 0.0700& \multicolumn{1}{c||}
{0.0700}& 1.0000& 0.6774& \multicolumn{1}{c|}
{0.6774}\\
\multicolumn{1}{|c||}{Wiki}& 0.9985& 0.9037& \multicolumn{1}{c||}
{0.9023}& 0.9981& 0.9346& \multicolumn{1}{c||}
{\textbf{0.9328}}& 0.9200& 0.9600& \multicolumn{1}{c||}
{0.8832}& 1.0000& 0.7011& \multicolumn{1}{c||}
{0.7011}& 0.9810& 0.3900& \multicolumn{1}{c||}
{0.3826}& 1.0000& 0.9004& \multicolumn{1}{c|}
{0.9004}\\
\multicolumn{1}{|c||}{Deezer}& 0.9996& 0.9730& \multicolumn{1}{c||}
{0.9726}& 0.9999& 0.9855& \multicolumn{1}{c||}
{\textbf{0.9854}}& 0.8500& 0.9900& \multicolumn{1}{c||}
{0.8415}& 1.0000& 0.3754& \multicolumn{1}{c||}
{0.3754}& 0.9750& 0.7400& \multicolumn{1}{c||}
{0.7215}& 1.0000& 0.9188& \multicolumn{1}{c|}
{0.9188}\\
\multicolumn{1}{|c||}{Slashdot}& 1.0000& 0.9904& \multicolumn{1}{c||}
{\textbf{0.9904}}& 1.0000& 0.9889& \multicolumn{1}{c||}
{0.9889}& 0.6320& 0.9900& \multicolumn{1}{c||}
{0.6257}& 0.7935& 0.9929& \multicolumn{1}{c||}
{0.7879}& 0.9900& 0.6200& \multicolumn{1}{c||}
{0.6138}& 1.0000& 0.9797& \multicolumn{1}{c|}
{0.9797}\\
\multicolumn{1}{|c||}{Twitter}& 1.0000& 0.9900& \multicolumn{1}{c||}
{\textbf{0.9900}}& 1.0000& 0.9893& \multicolumn{1}{c||}
{0.9893}& 0.6280& 0.9900& \multicolumn{1}{c||}
{0.6217}& 1.0000& 0.5542& \multicolumn{1}{c||}
{0.5542}& 0.9870& 0.4800& \multicolumn{1}{c||}
{0.4738}& 1.0000& 0.9740& \multicolumn{1}{c|}
{0.9740}\\
\multicolumn{1}{|c||}{DBLP}& 1.0000& 0.9954& \multicolumn{1}{c||}
{0.9954}& 1.0000& 0.9963& \multicolumn{1}{c||}
{\textbf{0.9963}}& 0.6460& 0.9900& \multicolumn{1}{c||}
{0.6395}& 1.0000& 0.1825& \multicolumn{1}{c||}
{0.1825}& 0.9930& 0.9200& \multicolumn{1}{c||}
{0.9136}& 1.0000& 0.8623& \multicolumn{1}{c|}
{0.8623}\\
\multicolumn{1}{|c||}{YouTube}& 1.0000& 0.9994& \multicolumn{1}{c||}
{\textbf{0.9994}}& 1.0000& 0.9987& \multicolumn{1}{c||}
{0.9987}& 0.5360& 0.9900& \multicolumn{1}{c||}
{0.5306}& 0.9990& 0.9705& \multicolumn{1}{c||}
{0.9695}& 0.9870& 0.7900& \multicolumn{1}{c||}
{0.7797}& 0.9936& 0.9968& \multicolumn{1}{c|}
{0.9904}\\
\multicolumn{1}{|c||}{Skitter}& 1.0000& 0.9995& \multicolumn{1}{c||}
{\textbf{0.9995}}& 1.0000& 0.9991& \multicolumn{1}{c||}
{0.9991}& 0.4270& 0.9900& \multicolumn{1}{c||}
{0.4227}& --& --& \multicolumn{1}{c||}
{--}& 0.9740& 0.7100& \multicolumn{1}{c||}
{0.6915}& 1.0000& 0.9884& \multicolumn{1}{c|}
{0.9884}\\
\hline
& \multicolumn{18}{c|}{\textbf{Influence Maximization}}  \\ \hline
\multicolumn{1}{|c||}{Facebook}& 0.9919& 0.7450& \multicolumn{1}{c||}
{0.7390}& 0.9208& 0.9027& \multicolumn{1}{c||}
{\textbf{0.8312}}& 0.9510& 0.7300& \multicolumn{1}{c||}
{0.6942}& 1.0039& 0.3248& \multicolumn{1}{c||}
{0.3261}& 0.9790& 0.0900& \multicolumn{1}{c||}
{0.0881}& 0.9938& 0.4917& \multicolumn{1}{c|}
{0.4887}\\
\multicolumn{1}{|c||}{Wiki}& 1.0269& 0.8831& \multicolumn{1}{c||}
{0.9069}& 0.9863& 0.9346& \multicolumn{1}{c||}
{\textbf{0.9218}}& 0.9690& 0.9000& \multicolumn{1}{c||}
{0.8721}& 0.9969& 0.6066& \multicolumn{1}{c||}
{0.6047}& 0.9600& 0.5100& \multicolumn{1}{c||}
{0.4896}& 1.0011& 0.8964& \multicolumn{1}{c|}
{0.8974}\\
\multicolumn{1}{|c||}{Deezer}& 0.9384& 0.9698& \multicolumn{1}{c||}
{0.9101}& 1.0262& 0.9855& \multicolumn{1}{c||}
{\textbf{1.0113}}& 0.8050& 0.5500& \multicolumn{1}{c||}
{0.4428}& 0.9837& 0.1093& \multicolumn{1}{c||}
{0.1075}& 0.9720& 0.7600& \multicolumn{1}{c||}
{0.7387}& 0.9984& 0.8478& \multicolumn{1}{c|}
{0.8464}\\
\multicolumn{1}{|c||}{Slashdot}& 0.9984& 0.9881& \multicolumn{1}{c||}
{0.9865}& 0.9959& 0.9889& \multicolumn{1}{c||}
{0.9848}& 0.9660& 0.9800& \multicolumn{1}{c||}
{0.9467}& 1.0076& 0.9875& \multicolumn{1}{c||}
{\textbf{0.9950}}& 0.9660& 0.7700& \multicolumn{1}{c||}
{0.7438}& 0.9993& 0.9797& \multicolumn{1}{c|}
{0.9790}\\
\multicolumn{1}{|c||}{Twitter}& 1.0045& 0.9965& \multicolumn{1}{c||}
{\textbf{1.0010}}& 0.9575& 0.9893& \multicolumn{1}{c||}
{0.9473}& 0.9200& 0.9800& \multicolumn{1}{c||}
{0.9016}& 0.9972& 0.4947& \multicolumn{1}{c||}
{0.4933}& 0.9660& 0.4000& \multicolumn{1}{c||}
{0.3864}& 0.9842& 0.9758& \multicolumn{1}{c|}
{0.9604}\\
\multicolumn{1}{|c||}{DBLP}& 0.8603& 0.9946& \multicolumn{1}{c||}
{0.8557}& 1.0495& 0.9963& \multicolumn{1}{c||}
{\textbf{1.0456}}& 0.8630& 0.9900& \multicolumn{1}{c||}
{0.8544}& 0.9873& 0.0194& \multicolumn{1}{c||}
{0.0192}& 0.9690& 0.8900& \multicolumn{1}{c||}
{0.8624}& 0.7558& 0.3325& \multicolumn{1}{c|}
{0.2513}\\
\multicolumn{1}{|c||}{YouTube}& 0.9677& 0.9994& \multicolumn{1}{c||}
{0.9671}& 0.9846& 0.9987& \multicolumn{1}{c||}
{0.9833}& 0.9330& 0.9900& \multicolumn{1}{c||}
{0.9237}& 0.9992& 0.9705& \multicolumn{1}{c||}
{0.9697}& 0.9710& 0.7500& \multicolumn{1}{c||}
{0.7282}& 1.0018& 0.9966& \multicolumn{1}{c|}
{\textbf{0.9984}}\\
\multicolumn{1}{|c||}{Skitter}& 0.9813& 0.9999& \multicolumn{1}{c||}
{0.9812}& 1.0032& 0.9991& \multicolumn{1}{c||}
{\textbf{1.0023}}& 0.8830& 0.9900& \multicolumn{1}{c||}
{0.8742}& --& --& \multicolumn{1}{c||}
{--}& 0.9830& 0.7800& \multicolumn{1}{c||}
{0.7667}& 1.0023& 0.9889& \multicolumn{1}{c|}
{0.9912}\\
\hline
\end{tabular}%
}
\end{table*}

\begin{figure}[t] 
  \centering
  \includegraphics[width=0.5\textwidth]{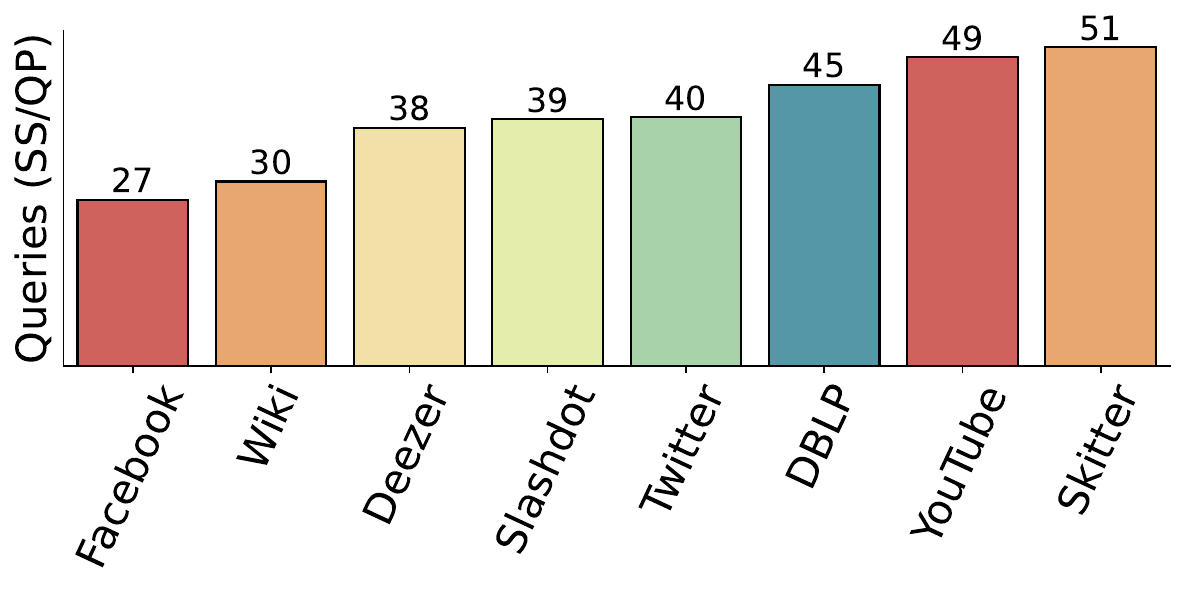}  
  \caption{ Comparison of the number of oracle calls between SS and \qs{}.}
  \label{fig:oracle_calls}

\end{figure}

In this section, we compare the empirical performance of \qs to
existing methods for the pruning problem across four different
objective functions, for both size and knapsack constraints.
The size constraint is the special case of the knapsack constraint
when $c(u) = 1$ for all $u \in \uni$ -- as the previous methods
for pruning are typically formulated for size constraint, this
allows for the fairest comparison with existing work. We provide the code and detailed instructions for reproducing the experiments in the supplementary material.

\textbf{Summary of results.} 
For both size and knapsack constraints,  we demonstrate
that \qs{} typically prunes over 90\% of the ground set
and sacrifices very little in terms of objective value,
achieving the highest combined metric (defined below) of
all pruning methods in
the majority of instances. 
The algorithm that is most competitive with \qs{} is
the Submodular Sparsification (SS) algorithm of  \citet{zhou2017scaling}.
However, SS requires $30$ times more oracle queries than \qs{}, as shown in Fig. \ref{fig:oracle_calls}.
For the ML-based pruning methods, a baseline \gnn that we introduce outperforms
the methods of
\gcomb{} \citep{manchanda2020gcomb}, \lense{} \citep{ireland2022lense}, and \comb{} \citep{tian2024combhelper}
on nearly every instance that we evaluated. 
Finally, in Section \ref{sec:multi_vs_single},
we show that  \qs{} improves over \qss{} by 5-10\% in objective value retained,
while increasing the size of the reduced ground set by less than 1\% percent. 

\textbf{Evaluation metrics.} We evaluate the algorithms using three performance metrics, where higher values indicate better performance in all cases. The first metric is the
\textit{pruning approximation ratio} $P_r$,
defined as the ratio of the objective value obtained from the pruned ground set $\uni'$ to the objective value from the original ground set $\uni$, both computed by a heuristic $\mathcal{H}$. Specifically, $P_r = f( \mathcal{H} (\uni')) / f( \mathcal{H} (\uni))$. The second metric is the
\textit{pruned fraction} $P_g$ of the original ground set that is pruned.
Finally, since an algorithm could trivially obtain the highest in one metric
by totally disregarding the other (pruning nothing or pruning everything),
we consider the \textit{combined metric}, $C = P_rP_g$.  


In Appendix  \ref{appendix:heuristics},  we summarize the heuristics used for each application, noting that our heuristic-agnostic algorithm is flexible for use with any heuristic. We use only a single attempt to prune the ground set for all algorithms due to the computational cost.

 \textbf{Applications.} 
We evaluate our algorithm on four applications: Maximum Cover (MaxCover), Maximum Cut (MaxCut), Influence Maximization (IM), and information retrieval. For detailed specifications of these applications, see Appendix \ref{appendix:problem_formulation}.

\textbf{Baselines and Prior Methods.} We compare our algorithm against algorithms that accelerate heuristics by pruning the ground set rather than directly optimizing the objective function. Our evaluation include the following classical and learning based methods: Submodular Sparsification (SS) \citep{zhou2017scaling}, \gcomb{} \citep{manchanda2020gcomb}, \lense{} \citep{ireland2022lense}, \comb{} \citep{tian2024combhelper}, and \gnn{}, a proposed baseline, which is a modified version of \comb{}. While the authors of \comb{} stated that they used GCN \citep{kipf2016semi} as their choice of GNN, we observed that the implementation of the algorithm actually employs SAGEConv \citep{hamilton2017inductive}. To ensure a fair comparison, we have included both GNN options in our evaluation. Empirically, we find that replacing SAGECONV  in \comb{} with GCN  significantly improves the performance of \comb{}, as GCN retains information from all neighbors, unlike SAGECONV. Hence, our proposed baseline, \gnn, a modified version of \comb{}, uses a GCN with fewer layers and  uses random numbers as node features, eliminating the need for domain knowledge and extensive feature engineering. Moreover,  \gnn also outperforms the other ML methods, \lense{} and \gcomb{}.
 
For the general knapsack constraint, the prior methods do not easily generalize. Since \gnn{} does easily generalize and
outperformed the other ML-based methods on size constraints, we only compare to \gnn for general knapsack constraints.
Instead of using random numbers as node features, we provide the degree, cost, and degree-to-cost ratio of each node as input features
for \gnn. Additionally, we compare to the \topk{} approach, which selects a set (where the size of the set is equal to the size of the pruned ground set by \qs{}) consisting of the top $k$ degree-to-cost ratios. Further descriptions of each algorithm can be found in Appendix \ref{appendix:baselines}.

\textbf{Datasets.} We evaluate our approach using real-world datasets from the Stanford Large Network Dataset Collection \citep{leskovec2016snap} for the traditional CO problems on graph, following the experimental setup of \cite{ireland2022lense}. For experiments related to the
retrieval system, we use the Beans \citep{beansdata}, CIFAR100 \citep{Krizhevsky09learningmultiple},  FOOD101 \citep{bossard14} and UCF101 dataset \citep{soomro2012ucf101dataset101human}. A summary of all datasets used can be found in Appendix \ref{appendix:dataset}.

\begin{figure*}[ht] 

  \subfigure{ 
    \includegraphics[width=0.23\textwidth]{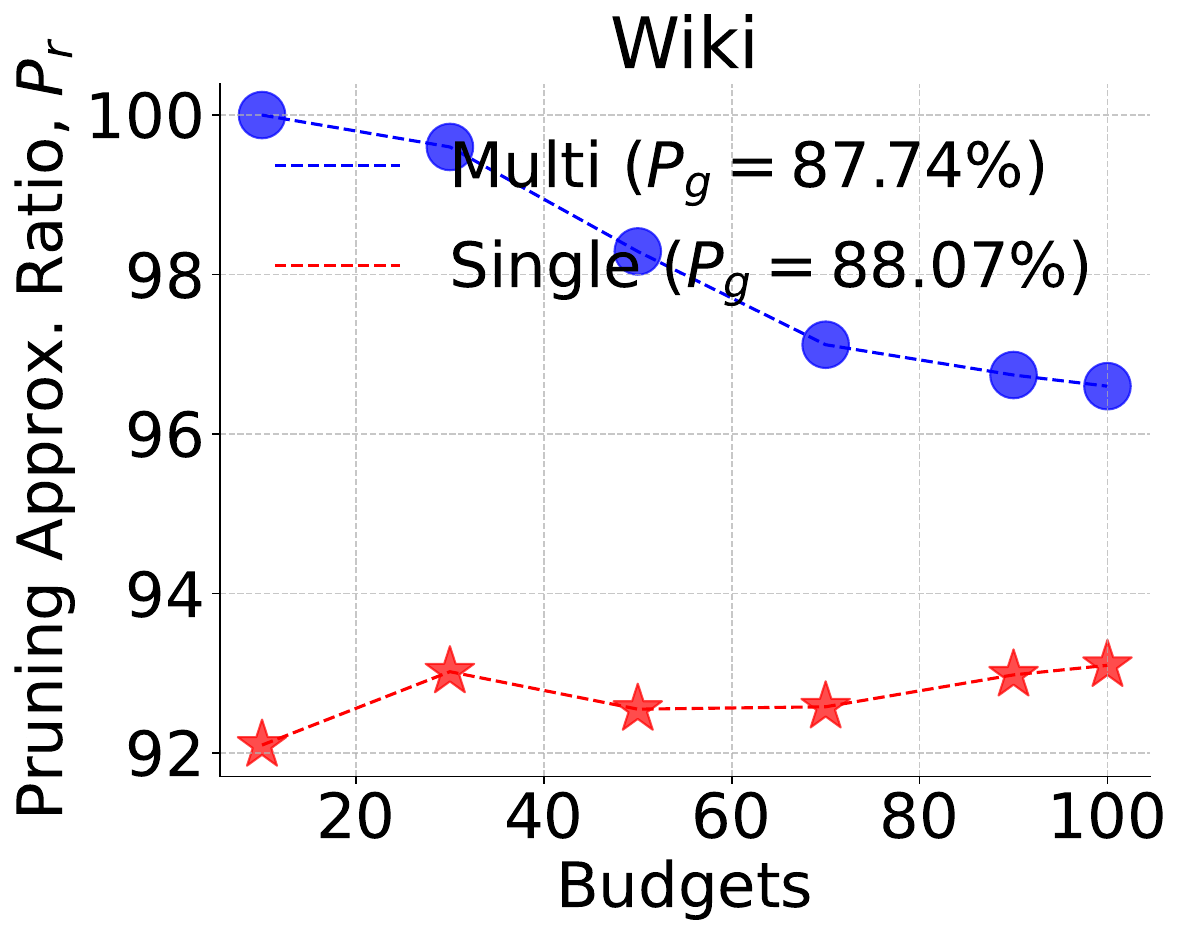}
  }
  \subfigure{ 
    \includegraphics[width=0.23\textwidth]{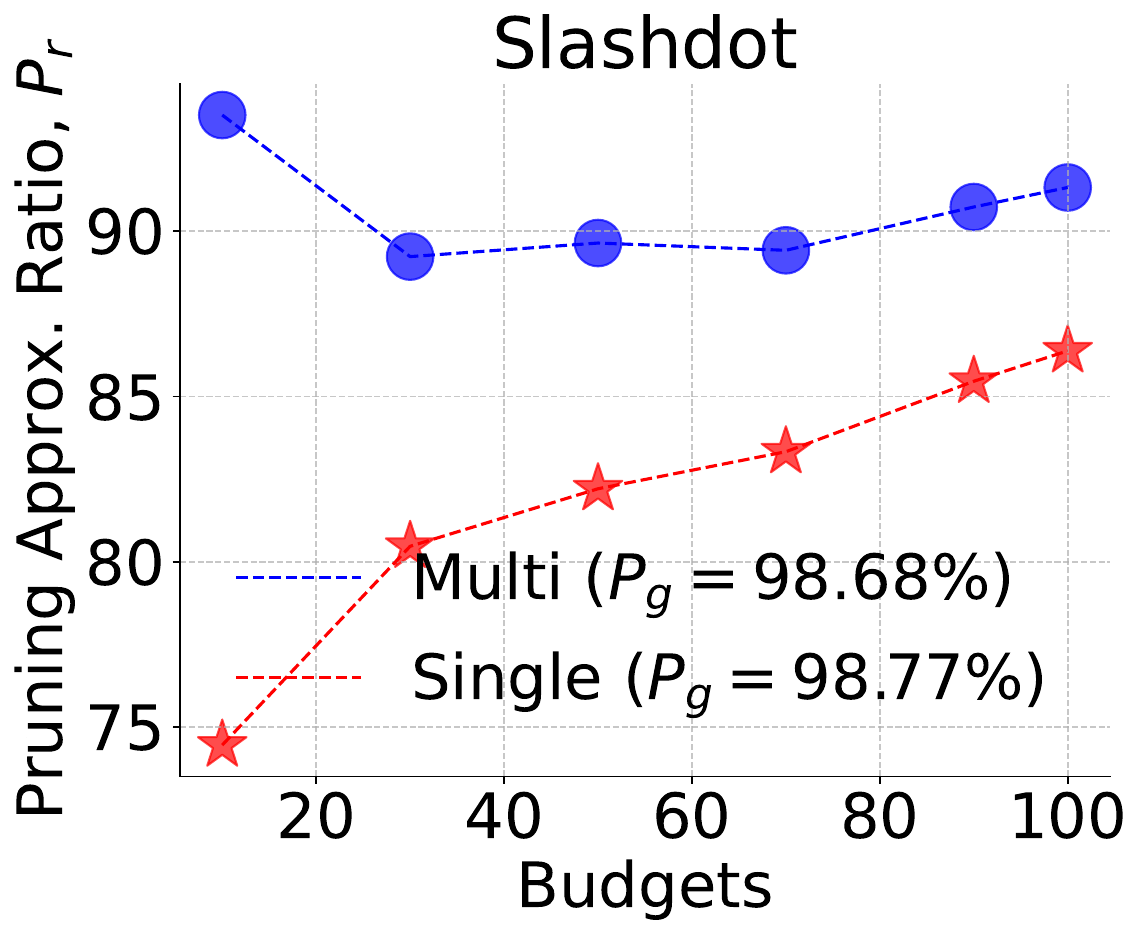}
  }
  \subfigure{ 
    \includegraphics[width=0.23\textwidth]{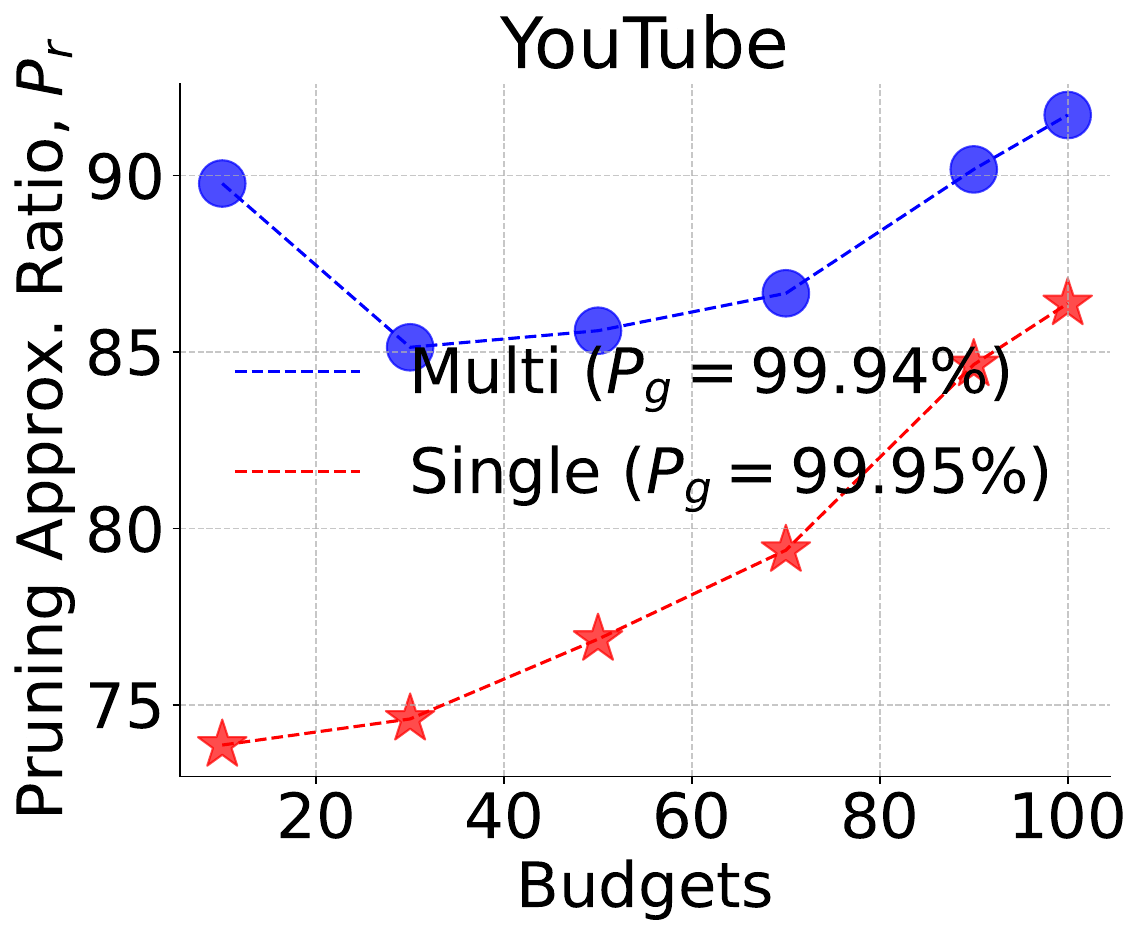}
  }
  \subfigure { 
    \includegraphics[width=0.23\textwidth]{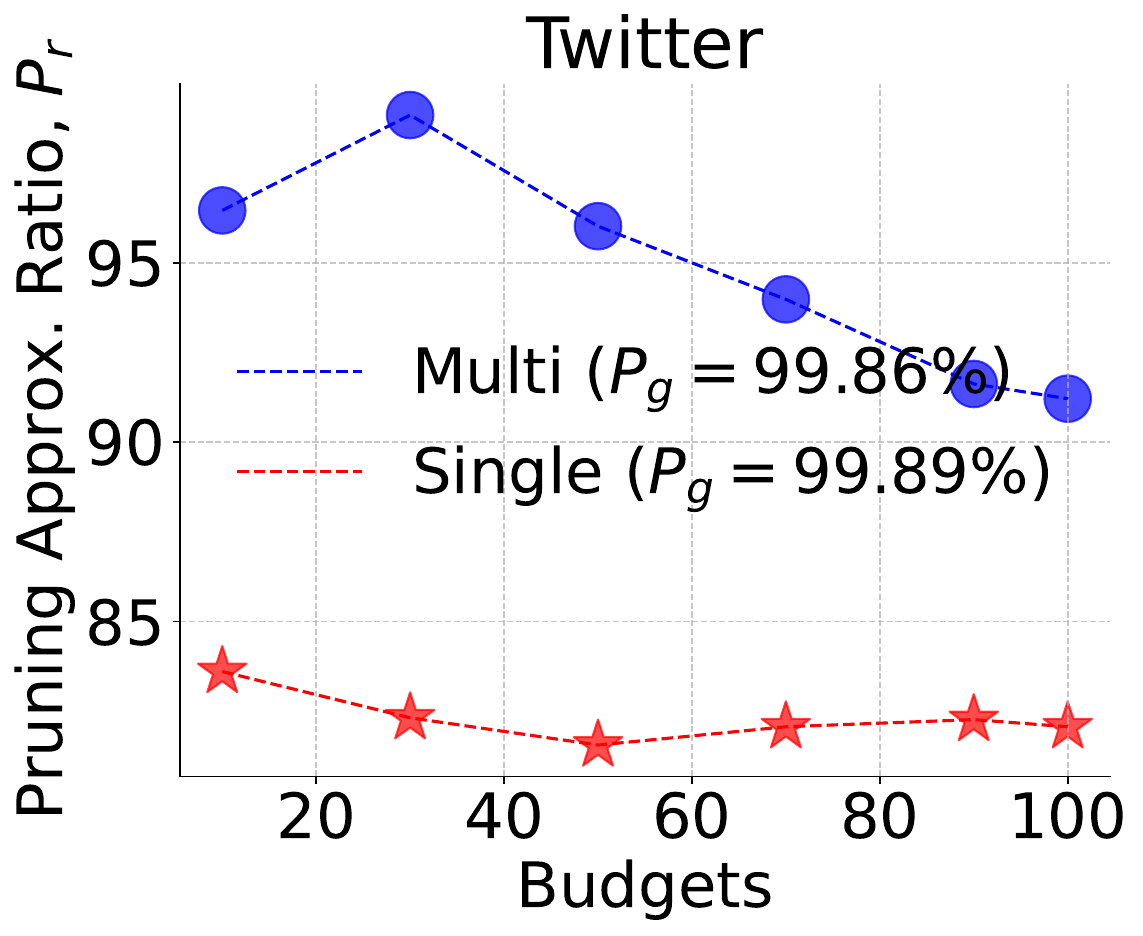}
  }
    
  \caption{Comparison between \qs and \qss  on a selection of data sets for IM.}
  \label{fig:main_paper_IM}
\end{figure*}
\subsection{Set 1: Evaluation on Size Constraints} \label{sec:size}
In this section, we evaluate the pruning methods for size constraints;
first for classical problems on graphs, and then for the image retrieval
application. 

\paragraph{Classical Problems on Graph.}
We consider traditional problems on graphs: Maximum Cover (MaxCover), Maximum Cut (MaxCut), and Influence Maximization (IM),
with a budget of $b=100$, following \cite{ireland2022lense}. From Table \ref{tab:size_constraint}, we observe that all algorithms achieve similar pruning approximation ratios. However, \qs{} is able to prune the original ground set by orders of magnitude more than most algorithms. Interestingly, and perhaps surprisingly, the IMM heuristic \citep{tang2015influence} occasionally achieves better solutions for IM
on the reduced ground set compared to the unpruned ground set (\eg Twitter, Wiki) -- which raises the intriguing
possibility of pruning to enhance the performance of a specific algorithm, which is out of the scope of this work. 

We observe that SS outperforms our method on some instances, particularly for IM.
However, as shown in Fig. \ref{fig:oracle_calls}, SS typically requires over $30$ times more oracle calls than \qs{}. Compared to other algorithms, \qs{} often surpasses these methods by orders of magnitude in pruning the ground set size while maintaining the same pruning approximation ratio. For example, \qs{} significantly outperforms \gcomb{} in pruning efficiency for MaxCover and MaxCut (we directly report the performance of \gcomb{} on these instances as reported in \cite{ireland2022lense}).

We remark that, in contrast to our method, the learning-based approaches frequently require domain knowledge and extensive feature engineering. For instance, \lense{} uses eigenvector centrality as a node feature, which is a costly calculation for large graphs and necessitates extensive hyperparameter tuning for each dataset and problem. Similarly, \comb{} employs problem-specific boosting, requiring domain knowledge.
Nevertheless, even the best-performing machine learning approaches are less optimal than our method and require significantly higher computational overhead. While LeNSE shows the lowest combined GPU and CPU memory usage among the ML approaches, \qs only utilizes CPU resources, which is less than the CPU usage of LeNSE alone. Please refer to Appendix \ref{appendix:scalability}.
\paragraph{Image Retrieval System Results.}

Next, we present our results for the image retrieval system under size constraint. We randomly sampled five images from the query dataset and averaged the performance over ten such queries. Fig. \ref{fig:retrival_system}, provides an overview of the pipeline and results for the retrieval system. We compare our algorithm with SS and \random{} (which selects a random set of elements from the candidate dataset, with a size equal to that of \qs{}). Notice that on the CIFAR100 dataset, SS fails to complete pruning within the three-hour cut-off time. We observe that \qs{} outperforms \random{} but fails to outperform SS. Although SS performs better than our approach on the FOOD101 dataset, it returns a ground set that is $14$ times larger than ours.
\subsection{Set 2: Performance on Knapsack Constraint}\label{sec:knapsack}
In this subsection, we compare our algorithm with the baselines for general knapsack constraints.

\textbf{Classical Problems on Graph.} We set the cost of each node following \cite{yaroslavtsev2020bring} (more details in Appendix \ref{appendix:problem_formulation}). In Table \ref{tab:knapsack_constraint}, we observe that \qs achieves the highest combined metric on 45\% of all instances tested. Specifically, for MaxCover and MaxCut, \qs often outperforms the simple baseline \topk{} by 20\% in retaining the objective value, except for Deezer-MaxCover and Skitter-MaxCut. The \topk{} approach does not account for the relationships between vertices, whereas \qs considers the problem more holistically, ensuring that the selected vertices collectively form a well-balanced subgraph. While \topk outperforms \qs for IM in most cases, its poor performance across MaxCut and MaxCover shows that \topk is not a robust approach across applications and datasets. On the other hand, \gnn performs quite well on some datasets but unexpectedly poorly on others.

\paragraph{Video Retrieval System Results.} We set the cost of each video proportional to its length.  For the video retrieval system, \gnn is not a suitable approach, as it cannot take a set of items as input. Therefore, we compare our approach with \random{} and \topk{} approach.
From Figure \ref{fig:video_summ}, we observe that \qs outperforms \random{} by 80\% while performing better than the \topk{} approach.

\subsection{Set 3: Multi-Budget and  Single-Budget Comparison} \label{sec:multi_vs_single}

Finally, we compare the performance of \qs and \qss (which prunes for the maximum budget) for a range of budget.
Following the multi-budget analysis in \cite{ireland2022lense}, we vary the budget $b$ from $10$ to $100$.
We remark that for size constraint, no significant improvement is observed for  \qs{} over \qss{} in the size-constrained experiments
-- we speculate that this may be at least partially explained by the efficacy of the standard greedy
algorithm for the size constraint, and a greedy solution of smaller size is automatically
contained in greedy solutions of larger sizes. 
However, for general knapsack constraint,
we observe that \qs  typically improves by 3\%-10\% over \qss{} for IM (Figure \ref{fig:main_paper_IM}), 
while only increasing the size of the ground set less than 1\% percent. \textcolor{black}{For MaxCover, we observe qualitatively similar
  improvements on some datasets; however, for MaxCut, \qs and \qss perform similarly.}
We provide results for each application and dataset in Appendix \ref{appendix:multi_vs_single}.


\section{CONCLUDING REMARKS AND FUTURE DIRECTIONS}  \label{sec:con}
In this work, we introduce \qs, the first theoretically principled
approach to pruning large ground sets for constrained discrete maximization
problems. Moreover, the theoretical guarantees extend to a range of budgets so
that the pruned ground set is applicable to more than a single optimization problem.
In addition to the theoretical properties, our algorithm exhibits excellent empirical
performance, and outperforms pre-existing heuristics for the pruning problem.
Future directions include improving the theoretical properties, as our ratio $\alpha$,
while constant, is rather small. Also, upper bounds on the pruning problem would shed
light on the complexity of the problem. Currently, we are unaware of existing hardness
results that would imply that even a ratio of $\alpha = 1$ is impossible.
Finally, extending the framework to handle more than one objective function, as well
as deletions to the ground set, would be an interesting avenue for future work. 
\clearpage


\bibliography{reference,pruning}

\onecolumn

\appendix

\section{Proofs for Section \ref{sec:alg}} \label{apx:proof}
\begin{proof}[Proof of Prop. \ref{prop:uprime-size}]
  Let $m^* = \left( 1 + \frac{\kappa }{\delta c_{min} } \right) \log ( n / \epsi ) + 1.$
  
  First, we show that at any time during the execution of Alg. \ref{alg:single}, $|A| \le m^* + |A_s|$.
  Fix a value $B$ of $A_s$, and let $A_i, 1 \le i \le m$ assume all distinct values of $A$ while $A_s$
  has value $B$, in the order in which they were assigned. Thus, $A_1 = A_s$, $A_2 = A_s + e_1, \ldots, A_m = A_s + e_1 + \ldots + e_{m-1}$.
  Then, let $(y_i = f(A_i))_{i=1}^m$. By Line \ref{line:add}, $y_i \ge \left( 1 + \frac{\delta c_{min}}{\kappa } \right) y_{i-1}$, for each $i \in [m]$. Let $\beta = \frac{\delta c_{min}}{\kappa }$; then by Fact \ref{fact:1}, if $i \ge m^* - 1$, then $y_i \ge \frac{n}{\epsi} y_1$, and a deletion is triggered on Line \ref{line:delete}, which changes the value of $A_s$. Hence, we have $m \le m^*$, which shows $|A| \le |A_s| = m^*$.

  Observe that after a deletion, $|A_s| = m \le m^*$. Since $A_s$ has initial value $\emptyset$, it holds that $|A_s| \le m^*$ throughout the execution of the algorithm. Therefore, $|A| \le 2m^*$, and thus $|\uni'| \le 2m^* +1$. 
\end{proof}
\begin{proof}[Proof of Prop. \ref{prop:ahat-adot}]
  Let $( \dot B_i )_{i=1}^m$ be the sequence of sets deleted by the algorithm on
  Line \ref{line:delete-set}, in the order in which they were deleted.
  Let $j(i)$ be the iteration of the \textbf{for} loop on which $B_i$
  was deleted, and let $A_{j(i)}$ be the value of the set $A$ at the
  beginning of iteration $j(i)$. Let $B_0 = \hat A$, and let
  $B_i = B_{i-1} \setminus \hat B_i$, for $i = 1$ to $m$.
  Observe that $B_m = \dot A$.

  We have that
  \begin{align*}
    f( B_{i-1} ) - f(B_i) &\overset{(a)}{\le} \gamma^{-1} f( \dot B_i ) \\
                          &\overset{(b)}{<} \frac{\gamma^{-1}\epsi}{n} f( A_{j(i)} ) \\
                          &\overset{(c)}{\le} \frac{\gamma^{-1}\epsi}{n} f( \hat A ),
  \end{align*}
  where Inequality (a) follows from $\gamma$-submodularity,
  Inequality (b) is from the condition to delete set $\dot B_i$ on Line \ref{line:delete},
  and Inequality (c) follows from monotonicity.
  From here,
  \begin{align*}
    f( \hat A ) - f( \dot A ) &= \sum_{i=1}^m f( B_{i-1} ) - f( B_i ) \\
    &\le m \cdot \frac{\gamma^{-1}\epsi}{n} f( \hat A ) \le \gamma^{-1}\epsi f( \hat A ). \qedhere
  \end{align*}
\end{proof}
\begin{proof}[Proof of Lemma \ref{lemma:value-inside-A}]
  If $c(A_m) \le \kappa$, let $A^* = A_m$ and there is nothing to show.
  Otherwise, let $A' =  \{a_{i'}, \ldots, a_m \}$, where $i' = \max_{i \in [m]} \{ i: c( \{ a_i, \ldots, a_m \} ) > \kappa \}$.
  We have
  \begin{align*}
    f(A') &\overset{(a)}{\ge} \gamma ( f(A_m) - f( A_m \setminus A') ) \\
          &\overset{(b)}{=} \gamma \sum_{i=i'}^m \marge{ a_i }{ A_{i-1} } \\
          &\overset{(c)}{\ge} \gamma^2 \sum_{i=i'}^m \frac{ \delta c(a_i) f( A_{i-1} )}{\kappa} \\
          &\overset{(d)}{>} \gamma^2 \delta f(A_{i'-1} ) = \gamma^2 \delta f( A \setminus A' ),
  \end{align*}
  where Inequality (a) follows from $\gamma$-submodularity and nonnegativity, Equality (b)
  is a telescoping sum, Inequality (c) is from the assumed condition (2), and Inequality
  (d) is from monotonicity and the fact that $c( A' ) > \kappa$.
  Thus
  \begin{align*}
    f(A_m) &\le \gamma^{-1} (f(A') + f( A_m \setminus A')) \\
           &< \gamma^{-1}(1 + 1/(\delta \gamma^2 )) f(A') = \frac{\delta \gamma^2  +1}{\delta \gamma^3 } f(A').
  \end{align*}
  Let $A'' = A' \setminus a_{i'}$; by choice of $A'$, it holds that $c(A'') \le \kappa$.
  Finally, from $\gamma$-submodularity,
  \begin{align*}
    f(a^*) + f(A'') &\ge f(a_{i'}) + f(A'') \\
                    &\ge \gamma f(A') \ge \frac{\delta \gamma^4}{\delta \gamma^2 + 1} f(A_m).
  \end{align*}
  Thus, set $A^* = \argmax \{ f( a^* ), f(A'') \}$. 
\end{proof}
\begin{proof}[Proof of Prof. \ref{prop:aprime}]
  Let $\dot A = \{\dot a_1, \ldots, \dot a_m \}$ be in the order in which the elements were added into set $A$.
  That is, if $j(i)$ the iteration of the \textbf{for} loop in which $\dot a_i$ is considered, $i < i'$ implies
  $j(i) < j(i')$. Let $A_j$ be the value of the set $A$ at the beginning of iteration $j$ of the \textbf{for}
  loop.
  Then
  \begin{align*}
    \marge{ \dot a_i }{ \dot A_{i-1} } &\overset{(a)}{\ge} \gamma \marge{ \dot a_i }{ A_{j(i)} } \\
                                       &\overset{(b)}{\ge} \gamma \frac{ \delta c( \dot a_i ) f( A_{j(i)} )}{\kappa} \\
                                       &\overset{(c)}{\ge} \gamma \frac{ \delta c( \dot a_i ) f( \dot A_{i - 1} )}{\kappa},
  \end{align*}
  where (a) follows from $\gamma$-submodularity since $\dot A_{i-1} \subseteq A_{j(i)}$,
  (b) follows from the condition on Line \ref{line:add},
  (c) follows from monotonicity. 
  Also $c(\dot a_i ) \le \kappa$, for all $i$. Therefore, $\dot A_i$ and $a^*$ satisfy the conditions of
  Lemma \ref{lemma:value-inside-A}, which implies the result. 
\end{proof}
\begin{proof}[Proof of Prop. \ref{prop:opt-ahat}]
  Let $O \subseteq \uni$, such that $c(O) \le \kappa$ and $f(O) = \opt$.
  For each $e \in \uni$, let $j(e)$ be the iteration of the \textbf{for}
  loop in which $e$ was considered, and let $A_j$ be the value of the set
  $A$ at the beginning of the $j$th iteration. 
  Then
  \begin{align*}
    \opt - f( \hat A ) &\overset{(a)}{\le } f( O \cup \hat A ) - f( \hat A) \\
                       &\overset{(b)}{\le } \gamma^{-1} \sum_{o \in O \setminus \hat A} \marge{o}{ A_{j(o)} } \\
                       &\overset{(c)}{<} \gamma^{-1} \sum_{o \in O \setminus \hat A} \frac{\delta c(o) f(A_{j(o)})}{\kappa } \\
                       &\overset{(d)}{\le} \gamma^{-1} \sum_{o \in O \setminus \hat A} \frac{\delta c(o) f( \hat A )}{\kappa } \\
                       &\overset{(e)}{\le} \gamma^{-1} \delta f( \hat A ),
  \end{align*}
  where (a) follows from monotonicity,
  (b) follows from $\gamma$-submodularity,
  (c) holds since the condition on Line \ref{line:add} must have failed since $o \not \in \hat A$,
  (d) follows from monotonicity, and (e) holds since $\sum_{o \in O} c(o) \le \kappa$.
\end{proof}
\begin{proof}[Proof of Prop. \ref{prop:part}] 
  Let $X \subseteq O$ be a maximal subset with $c(X) \le \kappa ( 1 - \eta )$.
  By Assumption \nhi, $X \neq \emptyset$. If $X = O$, there is nothing to show.
  Otherwise, let $o \in O \setminus X$. By definition, $c \left( O \setminus (X \cup \{ o \} \right) < \kappa - \kappa ( 1 - \eta ) = \eta \kappa < (1 - \eta) \kappa$, since $\eta \le 1/2$. Therefore, $O_1 = X$, $O_2 = \{ o \}$, and $O_3 = O \setminus (X \cup \{ o \} )$
  form the requisite partition. 
\end{proof}
\section{Experimental Setup} We run all our experiments on a Linux server equipped with an NVIDIA RTX A6000 GPU and an AMD EPYC $7713$ CPU, using PyTorch $2.4.1$ and Python $3.12.7$. Code and data are available at: \url{https://anonymous.4open.science/r/Pruning}.

\section{Problem Formulation}\label{appendix:problem_formulation}
In this section, we formally introduce the four applications mentioned in the paper. A CO problem on graph can be described with an undirected graph $G (V, E)$, where $V$ is the set of vertices, $E$ is the set of edges and each node $v \in V$ is associated with a cost $c(v)$. For knapsack constraint, we set the cost of each node $v \in V$ as $c(v) = \frac{\beta}{|V|} (|N(v)| - \alpha) $ where $N(v)$ is the set of neighbors of $v$ , $\alpha = \frac{1}{20}$ and $\beta$ is a normalizing factor so $c(v) \geq 1$,  so that the cost of each node is roughly proportional to the value of the node following \citet{yaroslavtsev2020bring}.

\textbf{Maximum Cover.} Given a budget $k$ , the goal of this problem is to find a subset of nodes $S \subseteq V$ that maximizes the objective function, $f(S)= |\{v |  v \in S \lor \exists(u,v) \in E , u \in S  \}|$ where $ \sum_{v \in S} c(v) \leq k$.

\textbf{Maximum Cut (MaxCut).} Given a budget $k$, the goal of this problem is to find a subset of nodes $S \subseteq V$ that maximizes the objective function, $f(S)= |\{ (u,v) \in E: v \in S, u \in V \setminus S \}|$ where $ \sum_{v \in S} c(v) \leq k$.

\textbf{Influence Maximization (IM).} Given a budget $k$, probabilities $p(u,v)$ on the edges $(u,v)\in E$ and a cascade model $\mathcal{C}$, the goal of this problem is to find a subset of nodes $S \subseteq V$ that maximizes the expected spread, $f(S)= \mathbf{E}[\sigma(S)]$ where $\sigma(S)$ denotes the spread of $S$ and $ \sum_{v \in S} c(v) \leq k$.  For the experiments, we consider the independent cascade model \citep{kempe2003maximizing} and set $p = 0.01$ following \citet{chen2009efficient}. Because relatively large propagation probability $p$ the influence spread is not very sensitive to different algorithms and heuristics, because a giant connected component exists even after removing every edge with probability $1-p$.

\textbf{Information Retrieval System.} Given a budget $k$, a set of candidate items $C$ (images or videos), and a set of query items $Q$, the objective is to find a subset $S \in C$ that maximizes the graph cut function, $f(S) = \lambda \sum_{i \in Q} \sum_{j \in S} s(i,j) - \sum_{i,j \in S} s(i,j) $  where $s$ is the similarity kernel and $\lambda \geq 2$. We set $\lambda = 10$ for our experiments. Note that the condition on $\lambda$ is to ensure that $f$ remains a monotone submodular function \cite{iyer2021submodular}. For video recommendation, the constraint is $\sum_{v \in S} L(v) \leq k$, where $L(v)$ represents the length of a video normalized the minimum length of the video. For the image retrieval system, the constraint is simply $|S| \leq k$.

\section{Baselines} \label{appendix:baselines}
In this section, we provide details about each algorithm discussed in our paper.

\textbf{GCOMB-P.} The idea is as follows: For a graph $G_{i} = (V, E)$ from training distribution and a budget of $b$, the vertices are initially sorted into descending order based on the sum of outgoing edge-weight (or degree in an unweighted graph), and $rank(v)$ denotes the position of vertex $v$ in this ordered list. A stochastic solver is used $m$ times to obtain $m$ different solution sets $ \{S^{(1)},S^{(2)},...,S^{(m)}\} $ for budget $b$. The goal here is to predict all nodes that could potentially be part of the solution set.  We define $r^{G_i}_{b}= \max_{v \in \cup_{j} S^{(j)}} rank(v)$ to be the highest rank of all vertices among the vertices in each of the $m$ solution sets. We repeat this process for each graph in the training distribution and define $r_{b} = \max_{G_i \in G_{train}} r^{G_i}_{b}$ for budget $b$. To generalize across budgets, we compute $r_b$ for a series of budgets and obtain $(b,r_b)$ pairs. To generalize across graph sizes, budgets and ranks are normalized with respect to the number of nodes in the graph. During testing, for an unseen budget $b$, linear interpolation is used to predict $r_b$.


\textbf{LeNSE.} \citet{ireland2022lense} proposed extracting a small portion of the original graph that most likely contains the optimal solution, where the solution can be extracted using any existing heuristics. The idea is as follows: Given a training distribution, extract the optimal solutions for each graph. Next, generate random subgraphs containing different portions of nodes from the optimal solution, where subgraphs are sets of nodes and their 1-hop neighbors. Assign these subgraphs to different classes based on the ratio of the objective function value from the subgraph to that of the original graph. The number of classes depends on the problem and dataset and usually requires domain knowledge.
In the next step, a reinforcement learning agent is trained to navigate a subgraph induced by a fixed number of random nodes and their 1-hop neighbors, updating the subgraph at each step by replacing a single vertex with its neighbors. By updating the subgraph in this manner, the agent aims to obtain a subgraph close to the optimal subgraph in the embedding space created by the encoder.

\textbf{Submodular Sparsification (SS).}   \citet{zhou2017scaling} proposed a randomized pruning method to reduce the submodular maximization problem via a novel concept called the submodularity graph. The submodularity graph is a weighted directed graph \( G(V,E,w) \) defined by a normalized submodular function \( f: 2^{V} \rightarrow \mathcal{R_+} \), where \( V \) is the set of nodes corresponding to the ground set, and each directed edge \( (u,v) \in E \) has weight \( w(u,v) = f(v \mid u) - f(u \mid V \setminus u) \). As computing all edge weights is quadratic in complexity, they proposed a randomized approach. The idea is as follows: Sample a subset of random nodes from the original ground set, remove these random nodes from the ground set at each step, and add them to the pruned ground set. Then, remove a subset of the top elements from the ground set that are deemed unimportant from the sample nodes in that set. After several stages of pruning, when the original ground set size falls below a certain threshold, the remaining elements are merged with the pruned ground set.

\textbf{COMBHelper.}  \citet{tian2024combhelper} introduced a method to train a GNN for vertex classification in combinatorial problems, predicting nodes likely to be part of the solution. To improve scalability, they applied knowledge distillation to transfer knowledge to a smaller GNN and used problem specific weight boosting to enhance the performance.


\textbf{GNNPruner.} We propose a modified version of \comb, in which we similarly train a vertex GNN classifier to predict solution nodes. Like \comb and \lense, we generate training labels using solutions from  a heuristic. However, unlike these methods, we use random numbers as node features for experiments under size constraints, and replace SAGEConv \citep{hamilton2017inductive} with a lightweight GCN \citep{kipf2016semi} with fewer layers. Empirical evidence shows that this architectural shift helps capture structural relationships in the graph without relying on the handcrafted features used by \comb and \lense. Additionally, it improves scalability by using fewer layers. We simplify prepossessing and improve prediction accuracy through more generalized learning by leveraging random features (shown in Table \ref{tab:size_constraint}). However, in experiments involving knapsack constraints, instead of using random numbers as node features, we provide the degree, cost, and the degree-to-cost ratio of each node as input features. This adjustment appears to improve performance.
\section{Dataset Information} \label{appendix:dataset}

In Table \ref{tab:dataset_distribution}, we list the graphs used in our empirical evaluations along with the sizes of the respective training and testing splits for the traditional CO problems on graphs. We follow \citet{ireland2022lense} to determine what percentage of the original graph edges was used for training and testing.

\begin{table}[h]
\centering
\caption{The real-world graphs used to perform our experiments. 
}

\vspace{1em}
\begin{tabular}{|c|c|c|c|c|}
\hline
         & \multicolumn{2}{c|}{Train Size}        & \multicolumn{2}{c|}{Test Size}           \\ \cline{2-5} 
Graph    & \multicolumn{1}{c|}{Vertices} & Edges  & \multicolumn{1}{c|}{Vertices} & Edges    \\ \hline
Facebook & 3847                          & $26470 (30\%)$  & 4002                          & 61764    \\
Wiki     & 4891                          & $30228 (30\%)$  & 6358                          & 70534    \\
Deezer   & 48870                         & $149460 (30\%)$& 53511                         & 348742   \\
SlashDot & 47546                         & $140566 (30\%)$& 67640                         & 327988   \\
Twitter  & 55827                         & $134229 (10\%)$& 80712                         & 1208067  \\
DBLP     & 63004                         & $41994 (10\%)$ & 315305                        & 1007872  \\
YouTube  & 185193                        & $179257 (06\%)$ & 1098104                       & 2808367  \\
Skitter  & 147604                        & $110952 (01\%)$ & 1694318                       & 10984346 \\ \hline
\end{tabular}

\label{tab:dataset_distribution}
\end{table}

For the information retrieval problem, we use the Beans dataset \citep{beansdata}, CIFAR100 \citep{Krizhevsky09learningmultiple}, FOOD101 \citep{bossard14}, and UCF101 \citep{soomro2012ucf101dataset101human}. Since the Beans dataset contains only three classes, we undersample one class to make the problem more challenging and use that class for querying. For the FOOD101 dataset, we sample 100 images per food category for our candidate set and query from the test dataset. Similarly, for the UCF-101 dataset, we select 100 videos for each action and randomly query from the validation dataset.

\section{Network architectures and hyper-parameters}

\paragraph{QuickPrune.} For all graph problem experiments, we set $\delta = 0.1$, $\epsilon = 0.1$, and $\eta = 0.5$. In the knapsack-constrained MaxCover experiment, we set $\delta = 0.5$. For both image and video retrieval systems, we use $\delta = 0.05$, $\epsilon = 0.1$, and $\eta = 0.5$.

\paragraph{SS.} For all experiments, we set $r=8$ and $c=8$, following \citet{zhou2017scaling}. We find that this configuration empirically performs well.

\paragraph{GNNPruner.} The architecture consists of two layers, each containing a graph convolutional layer with ReLU activation and $16 $ hidden channels. We use the Adam optimizer with a learning rate of $0.001$ and a weight decay of $5 \times 10^{-4}$, training with cross-entropy loss. At each epoch, we randomly sample vertices both from the solution and non-solution sets to prevent the loss from being dominated by the non-solution vertices.

\section{Efficiency and scalability analysis} \label{appendix:scalability}

From Figure \ref{fig:memory_comparsion} , we observe that most algorithms (\qs, SS, \gcomb) only utilize CPU resources. Algorithms like \lense, \comb, and \gnn use both CPU and GPU resources, with \comb showing the highest combined usage, making it the most computationally demanding algorithm , followed by GNNPruner; in contrast, LeNSE exhibits relatively minimal usage. LeNSE only modifies a subgraph, so it does not need to load the entire graph into the memory.

\begin{figure}
    \centering
    \includegraphics[width=0.5\linewidth]{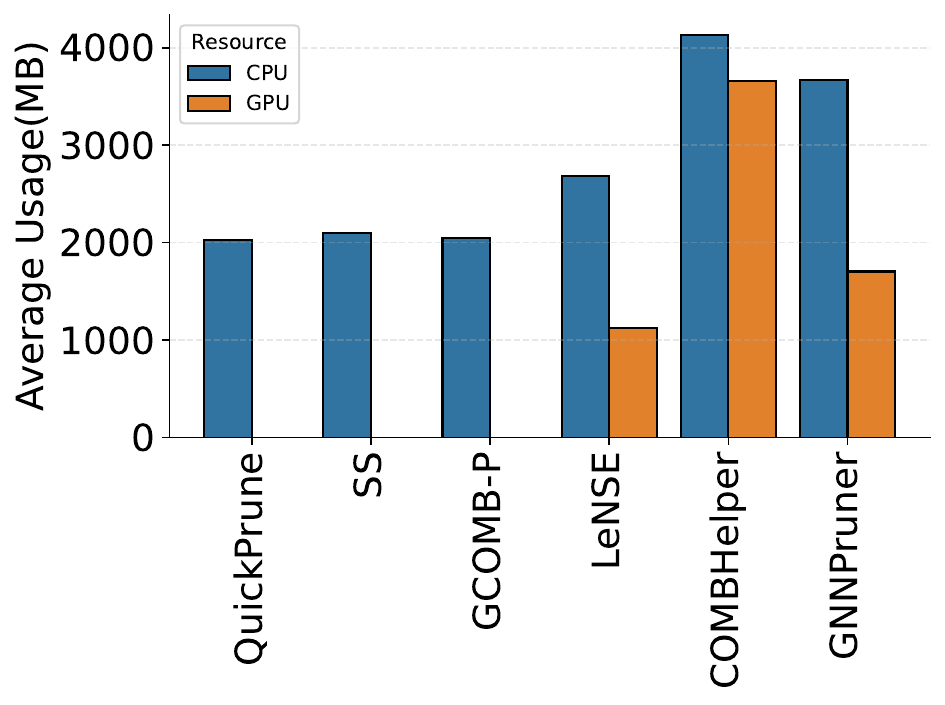}
    \caption{Average GPU and CPU memory utilization among
algorithms for MaxCover on YouTube dataset.}
    \label{fig:memory_comparsion}
\end{figure}

\section{Additional table and plots}

In this section, you can find the tables and plots omitted in the main paper due to space constraints.
\subsection{Heuristics}\label{appendix:heuristics}

In Table \ref{tab:dataset_distribution}, we summarize the algorithms used for each application under size and knapsack constraints. Note that for the knapsack-constrained experiments for IM, we use the improved greedy heuristic (referred to as Algorithm 2 in the original paper) from \citet{nguyen2013budgeted}.
\begin{table}[h]
\centering
\caption{Summary of Heuristic Approaches}
\vspace{1em}
\label{tab:heuristics}

\begin{tabular}{lll}
\toprule
Problem   & Size            & Knapsack                 \\ \midrule
MaxCover  & \citet{DBLP:journals/mp/NemhauserWF78}                & \citet{khuller1999budgeted} \\
MaxCut    & \citet{DBLP:journals/mp/NemhauserWF78}                 &   \citet{pham2023linear}   \\
IM        & \citet{tang2015influence} & \citet{nguyen2013budgeted} \\
Retrieval & \citet{DBLP:journals/mp/NemhauserWF78}                & \citet{khuller1999budgeted}                 \\ \bottomrule
\end{tabular}%

\end{table}

\newpage

\subsection{Image and Video Retrieval System.} \label{appendix:retrieval}

In Figure \ref{fig:retrival_system}, we present the results for information retrieval system. We employ the cosine similarity metric to assess the similarity between two images. Instead of using a generalist model trained on large datasets, we use a model fine-tuned on the specific dataset. This approach allows the underlying model to better understand the input images.

\begin{figure*}[h] \label{fig:retrival_system} 

  \subfigure []{ 
    \includegraphics[width=\textwidth]{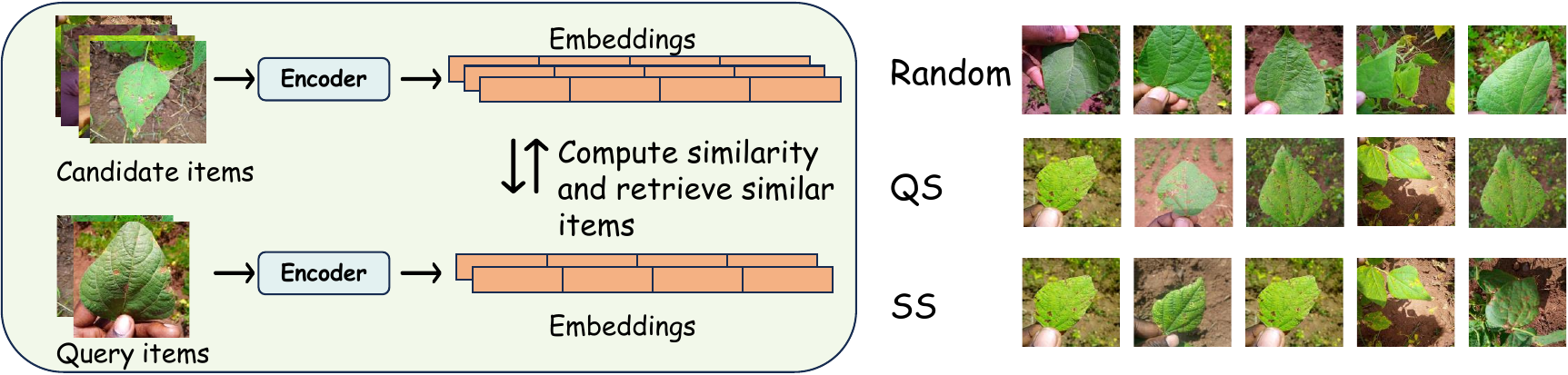}
  }
  \subfigure []{ 
    \includegraphics[width=0.23\textwidth]{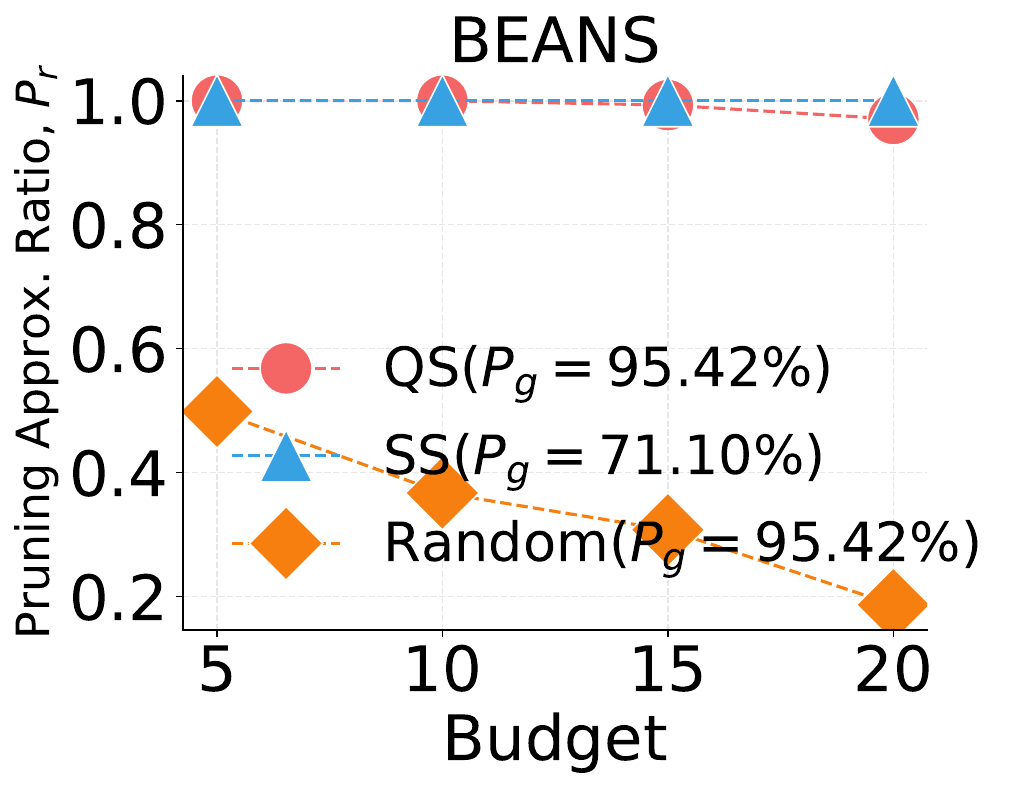}
  }
  \subfigure[]{ 
    \includegraphics[width=0.23\textwidth]{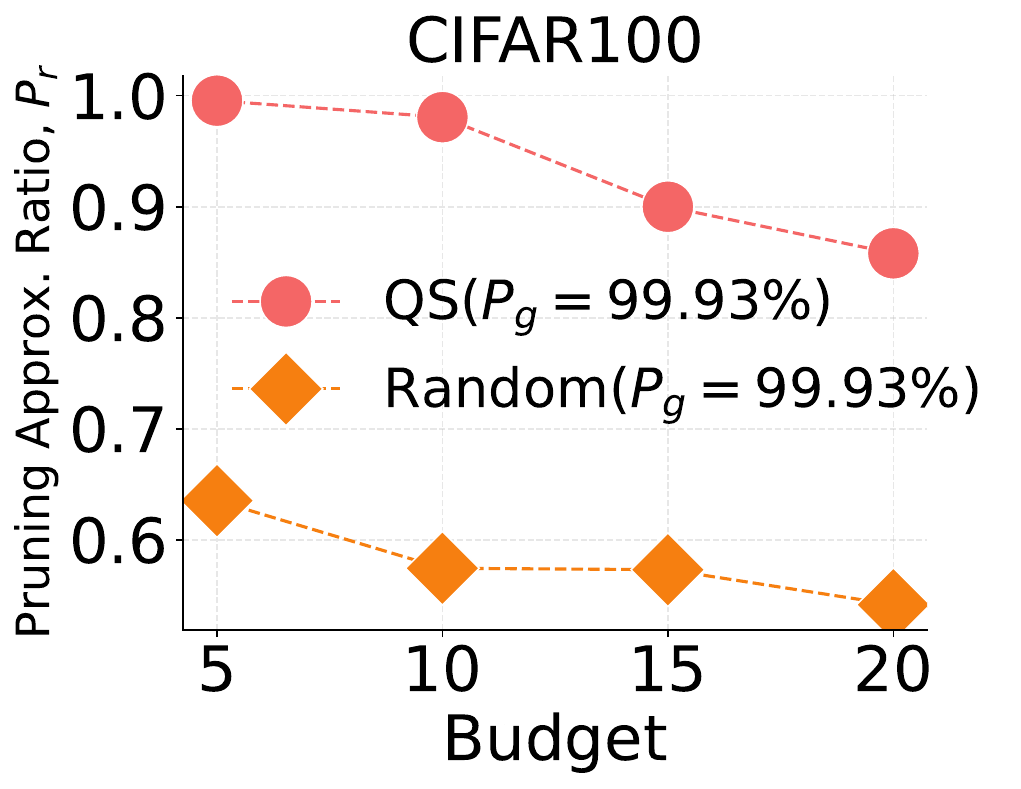}
  }
  \subfigure[]{ 
    \includegraphics[width=0.23\textwidth]{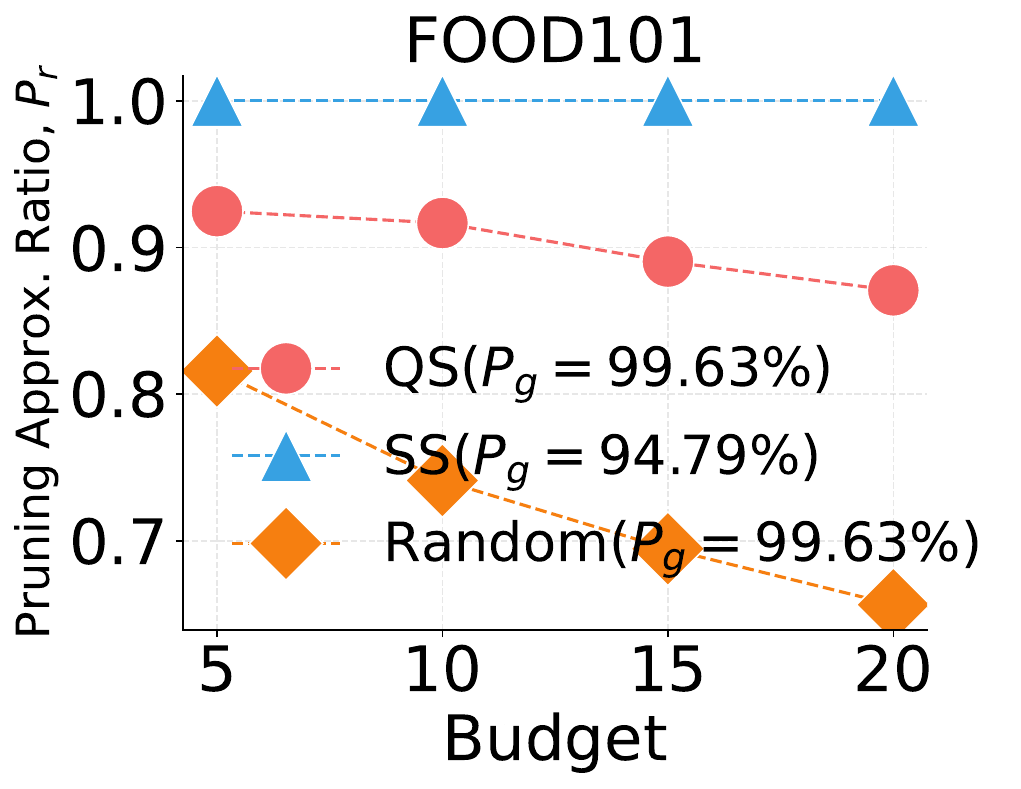}
  }
    \subfigure[]{ \label{fig:video_summ}
    \includegraphics[width=0.23\textwidth]{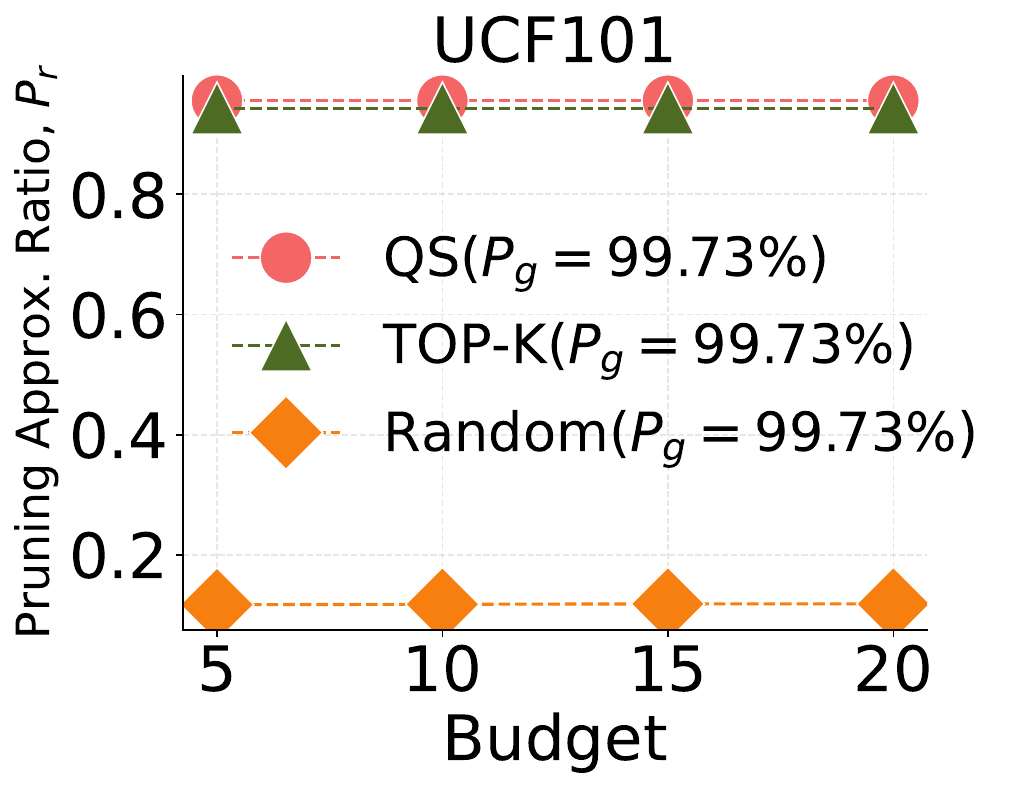}
  }

  \caption{Retrieval system: (a) The pipeline of the retrieval system and a visual representation of images selected by various algorithms for the Beans Dataset. (b-e) Multi-budget analysis of the retrieval system for budgets ranging from $5$ to $20$. Note that $P_g$ represents the percentage of the ground set that has been pruned.}

\end{figure*}

\subsection{Knapsack constrained experiments} \label{appendix:qs_vs_baselines_IM}

In Table \ref{tab:knapsack_constraint},  we present the results of our experiments under knapsack constraint.

\begin{table}[h]
\centering
\caption{Comparison of pruning algorithms for knapsack constraint experiments (best combined metric in bold).}
\label{tab:knapsack_constraint}
\vspace{1em}
\begin{tabular}{|cccccccccc|}
\hline
\multicolumn{1}{|c|}{}         & \multicolumn{3}{c|}{\qs}                                              & \multicolumn{3}{c|}{\topk}                                           & \multicolumn{3}{c|}{\gnn}                       \\ \hline
\multicolumn{1}{|c|}{Graph}    & $P_r \uparrow$ & $P_g$ $\uparrow$ & \multicolumn{1}{c|}{$C \uparrow$} & $P_r \uparrow$ & $P_g$ $\uparrow$ & \multicolumn{1}{c|}{$C\uparrow$} & $P_r \uparrow$ & $P_g$ $\uparrow$ & $C\uparrow$ \\ \hline
& \multicolumn{9}{c|}{\textbf{Maximum Cover}}  \\ \hline
 \multicolumn{1}{|c|}{Facebook} 
& 0.9725& 0.9871& \multicolumn{1}{c|}
{\textbf{0.9600}}& 0.5505& 0.9871& \multicolumn{1}{c|}
{0.5434}& 0.3303& 0.9936& \multicolumn{1}{c|}
{0.3282} \\
 \multicolumn{1}{|c|}{Wiki} 
& 0.9400& 0.9851& \multicolumn{1}{c|}
{\textbf{0.9260}}& 0.7550& 0.9851& \multicolumn{1}{c|}
{0.7438}& 1.0000& 0.8696& \multicolumn{1}{c|}
{0.8696} \\
 \multicolumn{1}{|c|}{Deezer} 
& 0.8750& 0.9983& \multicolumn{1}{c|}
{0.8735}& 0.9150& 0.9983& \multicolumn{1}{c|}
{0.9134}& 1.0000& 0.9725& \multicolumn{1}{c|}
{\textbf{0.9725}} \\
 \multicolumn{1}{|c|}{Slashdot} 
& 0.8950& 0.9988& \multicolumn{1}{c|}
{\textbf{0.8939}}& 0.5700& 0.9988& \multicolumn{1}{c|}
{0.5693}& 1.0000& 0.8325& \multicolumn{1}{c|}
{0.8325} \\
 \multicolumn{1}{|c|}{Twitter} 
& 0.8700& 0.9987& \multicolumn{1}{c|}
{\textbf{0.8689}}& 0.5950& 0.9987& \multicolumn{1}{c|}
{0.5942}& 1.0000& 0.4783& \multicolumn{1}{c|}
{0.4783} \\
 \multicolumn{1}{|c|}{DBLP} 
& 0.8450& 0.9997& \multicolumn{1}{c|}
{0.8447}& 0.7150& 0.9997& \multicolumn{1}{c|}
{0.7148}& 1.0000& 0.9971& \multicolumn{1}{c|}
{\textbf{0.9971}} \\
 \multicolumn{1}{|c|}{YouTube} 
& 0.9150& 0.9999& \multicolumn{1}{c|}
{\textbf{0.9149}}& 0.5200& 0.9999& \multicolumn{1}{c|}
{0.5199}& 1.0000& 0.3626& \multicolumn{1}{c|}
{0.3626} \\
 \multicolumn{1}{|c|}{Skitter} 
& 0.9050& 0.9999& \multicolumn{1}{c|}
{0.9049}& 0.8150& 0.9999& \multicolumn{1}{c|}
{0.8149}& 1.0000& 0.9887& \multicolumn{1}{c|}
{\textbf{0.9887}} \\
\hline
& \multicolumn{9}{c|}{\textbf{Maximum Cut}}  \\ \hline
 \multicolumn{1}{|c|}{Facebook} 
& 0.9899& 0.9666& \multicolumn{1}{c|}
{0.9568}& 1.0000& 0.9666& \multicolumn{1}{c|}
{0.9666}& 0.9697& 0.4984& \multicolumn{1}{c|}
{0.4833} \\
 \multicolumn{1}{|c|}{Wiki} 
& 0.9600& 0.9928& \multicolumn{1}{c|}
{\textbf{0.9531}}& 0.5100& 0.9928& \multicolumn{1}{c|}
{0.5063}& 1.0000& 0.1836& \multicolumn{1}{c|}
{0.1836} \\
 \multicolumn{1}{|c|}{Deezer} 
& 0.9600& 0.9985& \multicolumn{1}{c|}
{0.9586}& 0.8000& 0.9985& \multicolumn{1}{c|}
{0.7988}& 1.0000& 0.9745& \multicolumn{1}{c|}
{\textbf{0.9745}} \\
 \multicolumn{1}{|c|}{Slashdot} 
& 0.9700& 0.9991& \multicolumn{1}{c|}
{\textbf{0.9691}}& 0.6700& 0.9991& \multicolumn{1}{c|}
{0.6694}& 1.0000& 0.4395& \multicolumn{1}{c|}
{0.4395} \\
 \multicolumn{1}{|c|}{Twitter} 
& 0.9500& 0.9996& \multicolumn{1}{c|}
{0.9496}& 0.2900& 0.9996& \multicolumn{1}{c|}
{0.2899}& 1.0000& 0.0222& \multicolumn{1}{c|}
{0.0222} \\
 \multicolumn{1}{|c|}{DBLP} 
& 0.9500& 0.9999& \multicolumn{1}{c|}
{0.9499}& 0.4300& 0.9999& \multicolumn{1}{c|}
{0.4300}& 1.0000& 0.9988& \multicolumn{1}{c|}
{\textbf{0.9988}} \\
 \multicolumn{1}{|c|}{YouTube} 
& 0.9700& 0.9999& \multicolumn{1}{c|}
{\textbf{0.9699}}& 0.7300& 0.9999& \multicolumn{1}{c|}
{0.7299}& 1.0000& 0.9406& \multicolumn{1}{c|}
{0.9406} \\
 \multicolumn{1}{|c|}{Skitter} 
& 0.9700& 0.9999& \multicolumn{1}{c|}
{0.9699}& 1.0000& 0.9999& \multicolumn{1}{c|}
{\textbf{0.9999}}& 1.0000& 0.9993& \multicolumn{1}{c|}
{0.9993} \\
\hline
& \multicolumn{9}{c|}{\textbf{Influence Maximization}}  \\ \hline
 \multicolumn{1}{|c|}{Facebook} 
& 0.9934& 0.9148& \multicolumn{1}{c|}
{\textbf{0.9088}}& 0.9934& 0.9148& \multicolumn{1}{c|}
{\textbf{0.9088}}& 0.1173& 0.9854& \multicolumn{1}{c|}
{0.1156} \\
 \multicolumn{1}{|c|}{Wiki} 
& 0.9660& 0.8774& \multicolumn{1}{c|}
{0.8476}& 1.0000& 0.8774& \multicolumn{1}{c|}
{\textbf{0.8774}}& 0.8280& 0.8491& \multicolumn{1}{c|}
{0.7031} \\
 \multicolumn{1}{|c|}{Deezer} 
& 1.0000& 0.9820& \multicolumn{1}{c|}
{\textbf{0.9820}}& 1.0000& 0.9820& \multicolumn{1}{c|}
{\textbf{0.9820}}& 0.4678& 0.9719& \multicolumn{1}{c|}
{0.4547} \\
 \multicolumn{1}{|c|}{Slashdot} 
& 0.9132& 0.9868& \multicolumn{1}{c|}
{0.9011}& 1.0120& 0.9868& \multicolumn{1}{c|}
{\textbf{0.9986}}& 0.7385& 0.8403& \multicolumn{1}{c|}
{0.6206} \\
 \multicolumn{1}{|c|}{Twitter} 
& 0.9121& 0.9986& \multicolumn{1}{c|}
{0.9108}& 0.9689& 0.9986& \multicolumn{1}{c|}
{\textbf{0.9675}}&  0.1399 & 0.9999 & \multicolumn{1}{c|}
{0.1398} \\
 \multicolumn{1}{|c|}{DBLP} 
& 1.0000& 0.9958& \multicolumn{1}{c|}
{0.9958}& 1.0000& 0.9958& \multicolumn{1}{c|}
{0.9958}& 1.0000& 0.9999& \multicolumn{1}{c|}
{\textbf{0.9999}} \\
 \multicolumn{1}{|c|}{YouTube} 
& 0.9171& 0.9994& \multicolumn{1}{c|}
{0.9165}& 0.9502& 0.9994& \multicolumn{1}{c|}
{\textbf{0.9496}}& 0.9814& 0.3299& \multicolumn{1}{c|}
{0.3238} \\
\multicolumn{1}{|c|}{Skitter}  &    0.8291            &     0.9999             &           \multicolumn{1}{c|}
{0.8290}                        &   0.9391             &    0.9999               &    \multicolumn{1}{c|}
{\textbf{0.9390}}                              &    0.9999            &    0.1399              &     0.1398        \\ \hline
\end{tabular}%

\end{table}

\subsection{Comparison between \qs and \qss} \label{appendix:multi_vs_single}

In this section, we present the comparison between \qs and \qss under knapsack constraint. We run \qss with the maximum budget in the range (in our case, this would be 100).

\begin{figure*}[] 
  \subfigure { 
    \includegraphics[width=0.23\textwidth]{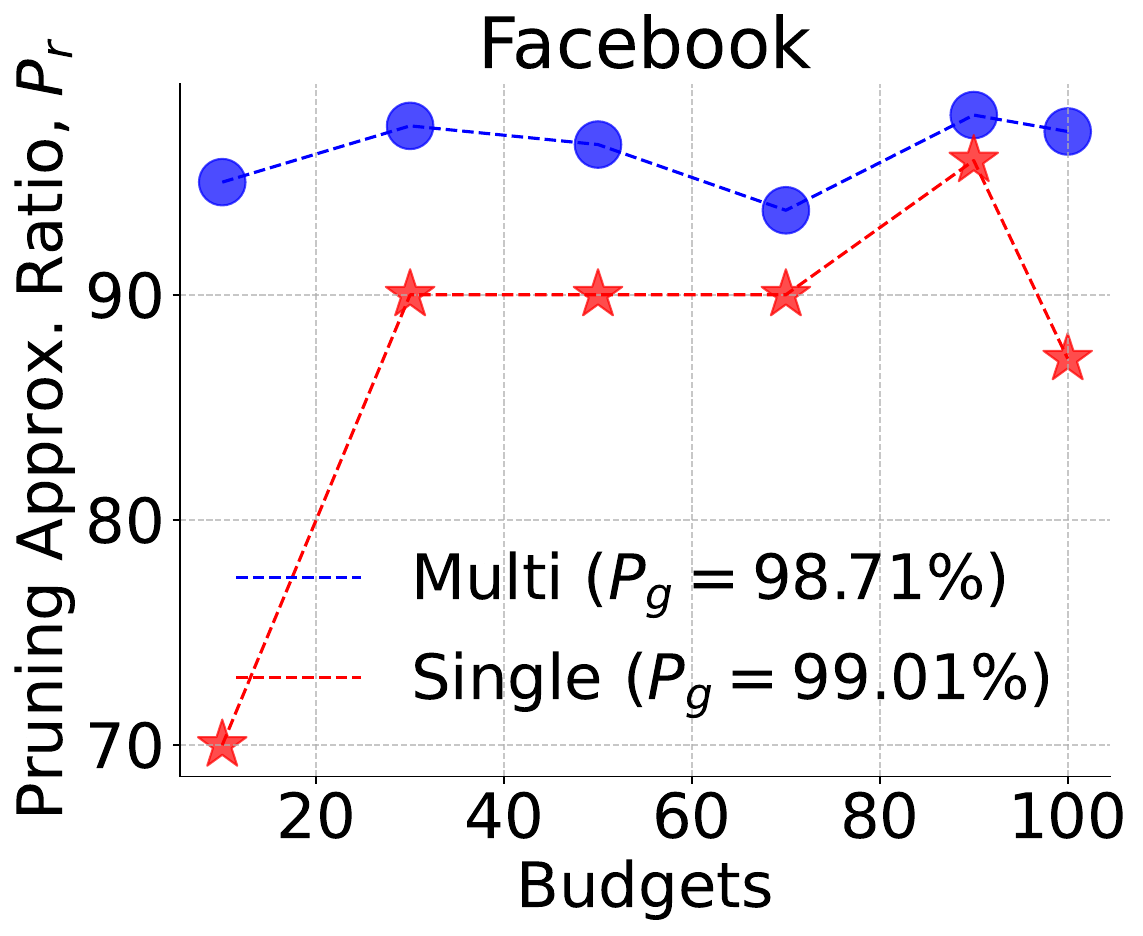}
  }
  \subfigure{ 
    \includegraphics[width=0.23\textwidth]{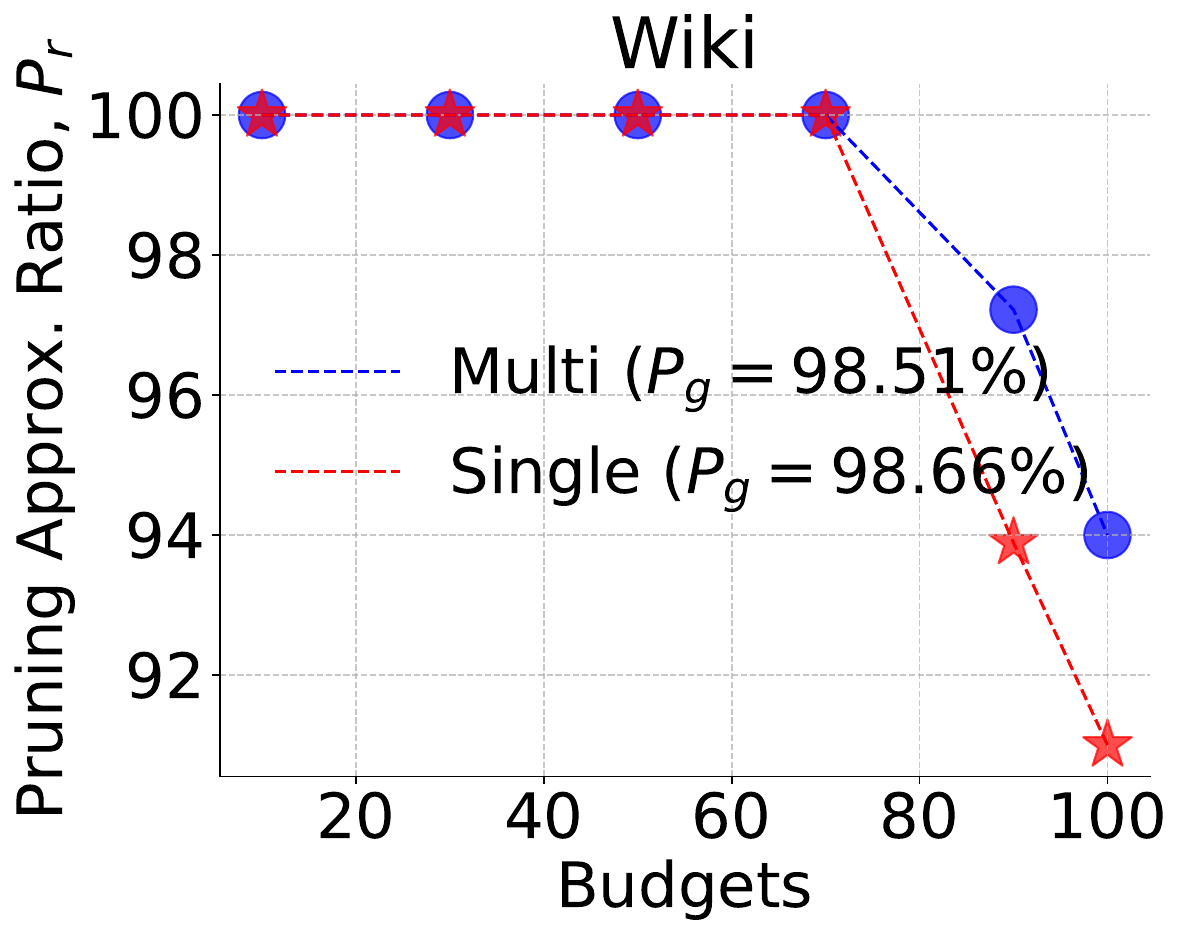}
  }
  \subfigure{ 
    \includegraphics[width=0.23\textwidth]{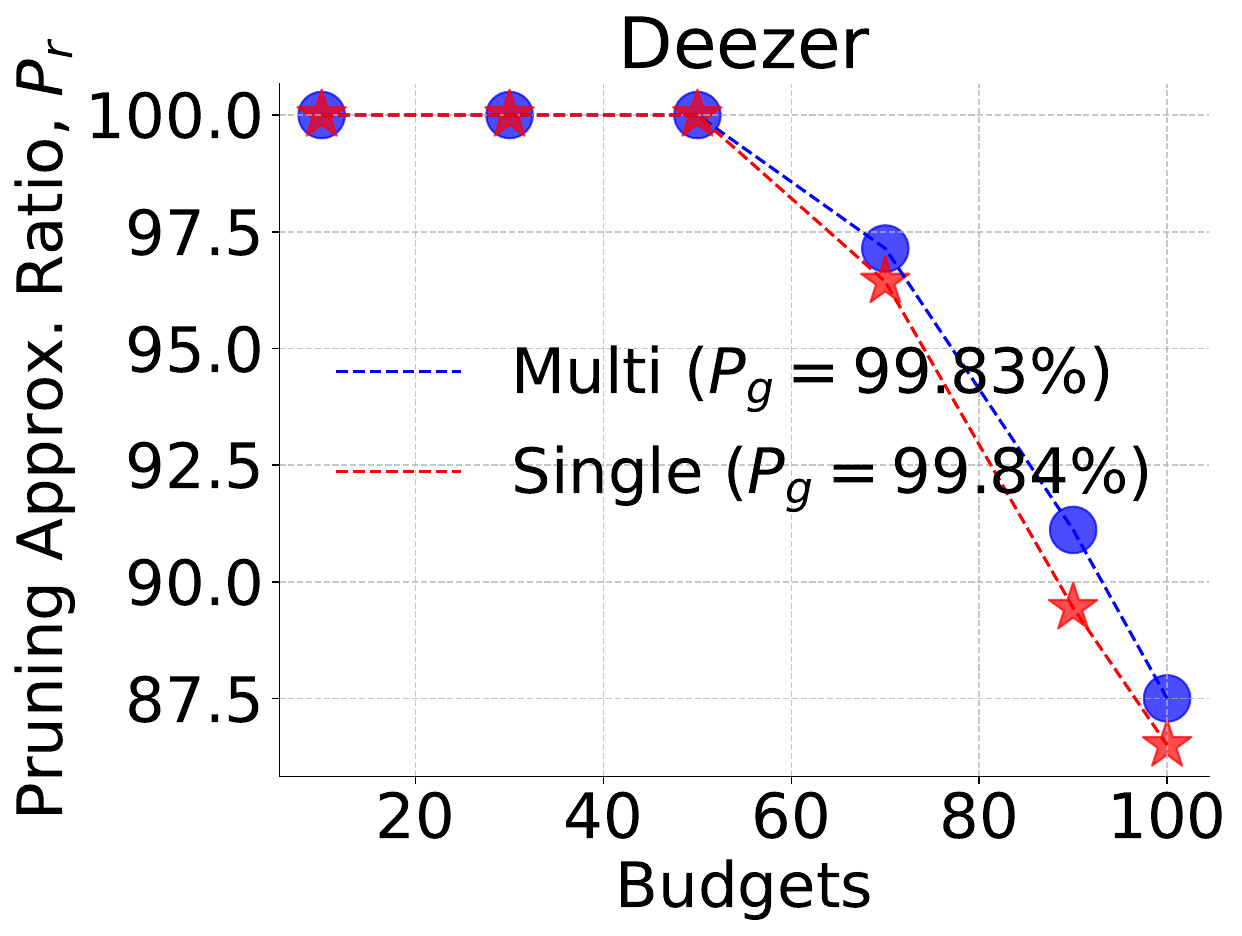}
  }
    \subfigure{ 
    \includegraphics[width=0.23\textwidth]{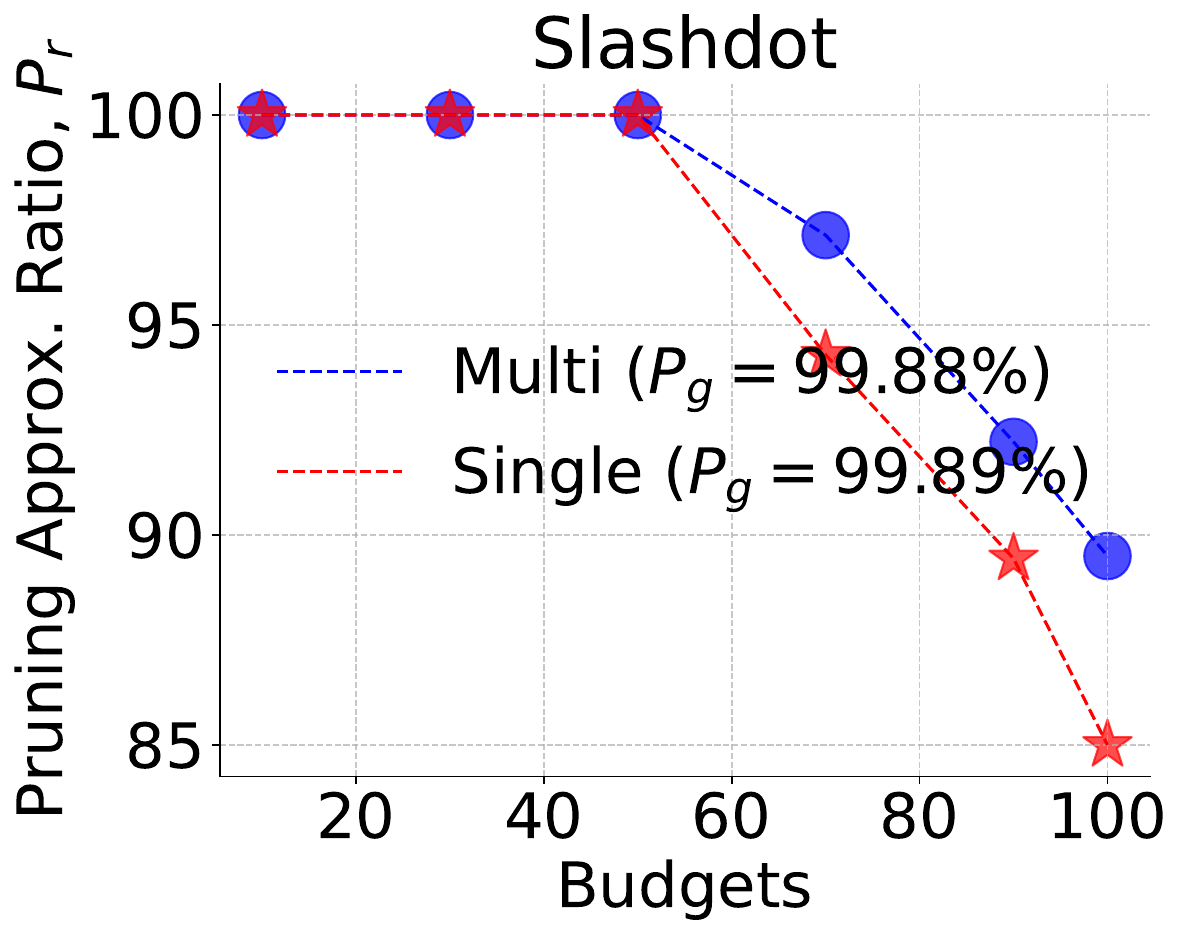}
  }
  \subfigure { 
    \includegraphics[width=0.23\textwidth]{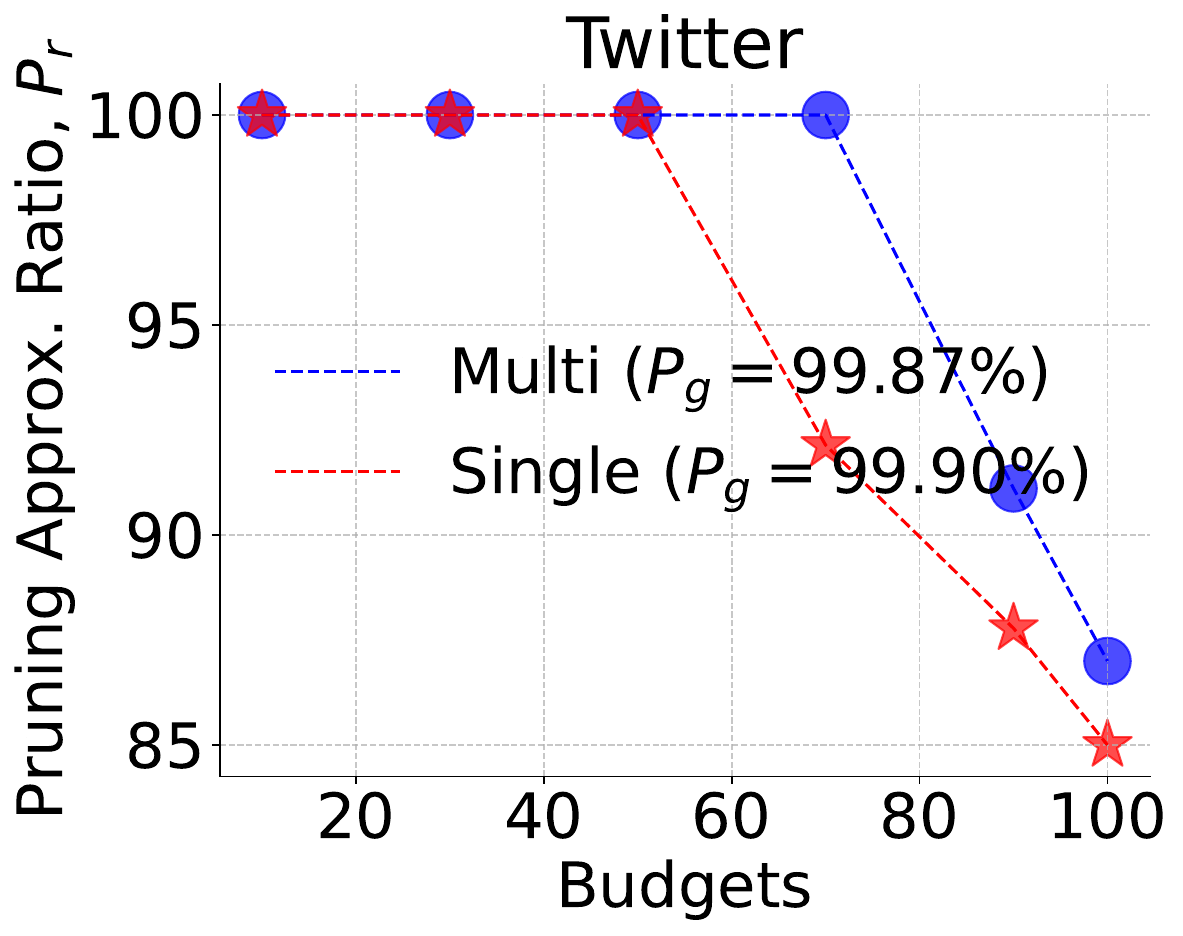}
  }
  \subfigure{ 
    \includegraphics[width=0.23\textwidth]{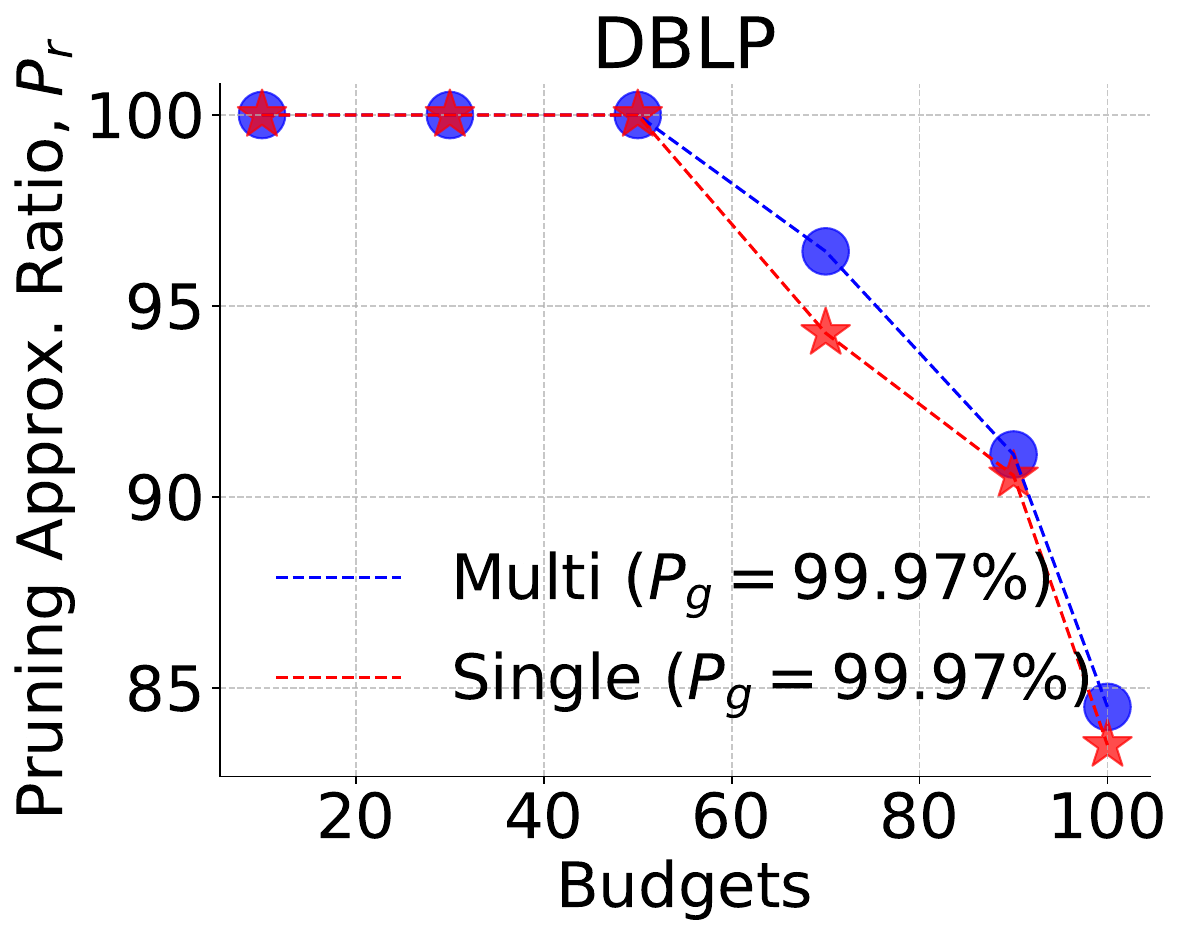}
  }
  \subfigure{ 
    \includegraphics[width=0.23\textwidth]{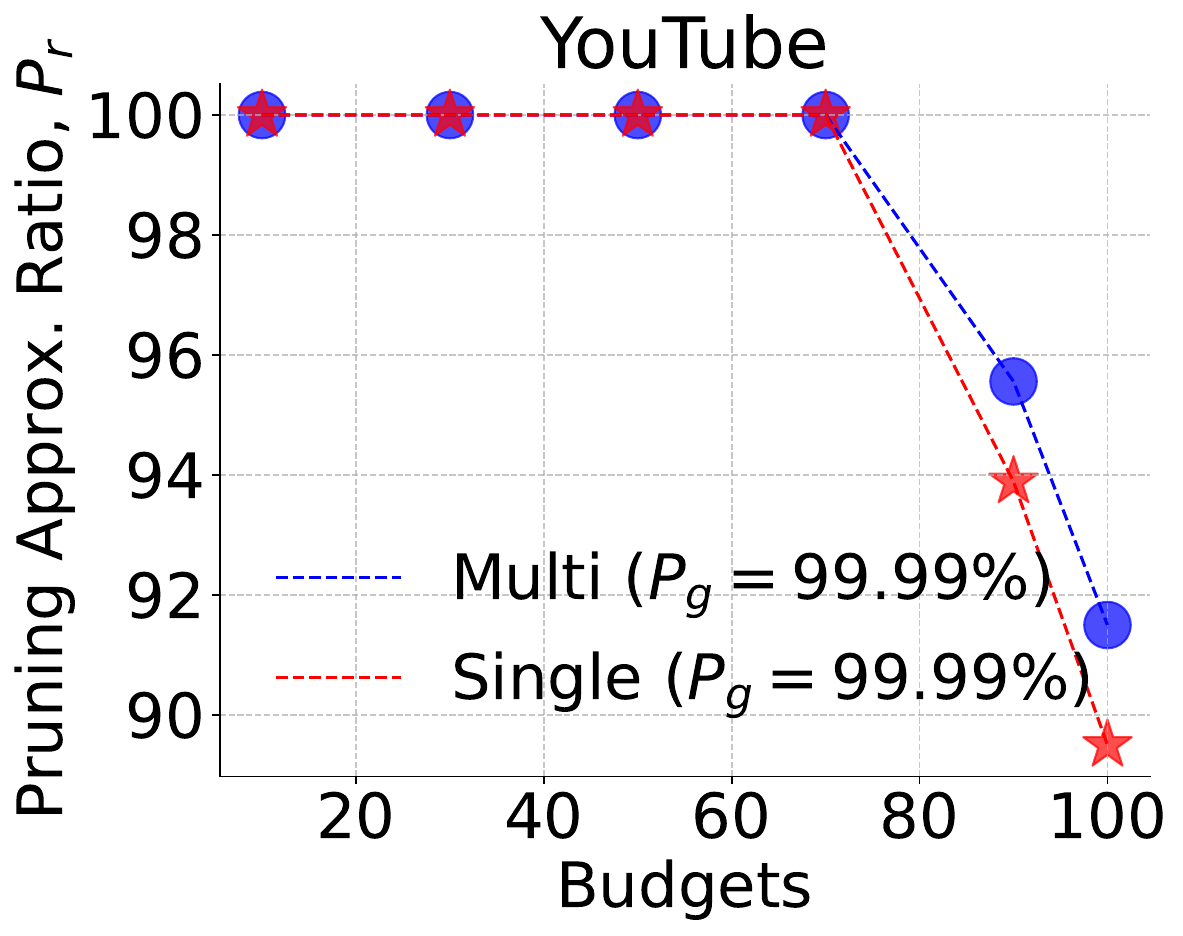}
  }
    \subfigure{ 
    \includegraphics[width=0.23\textwidth]{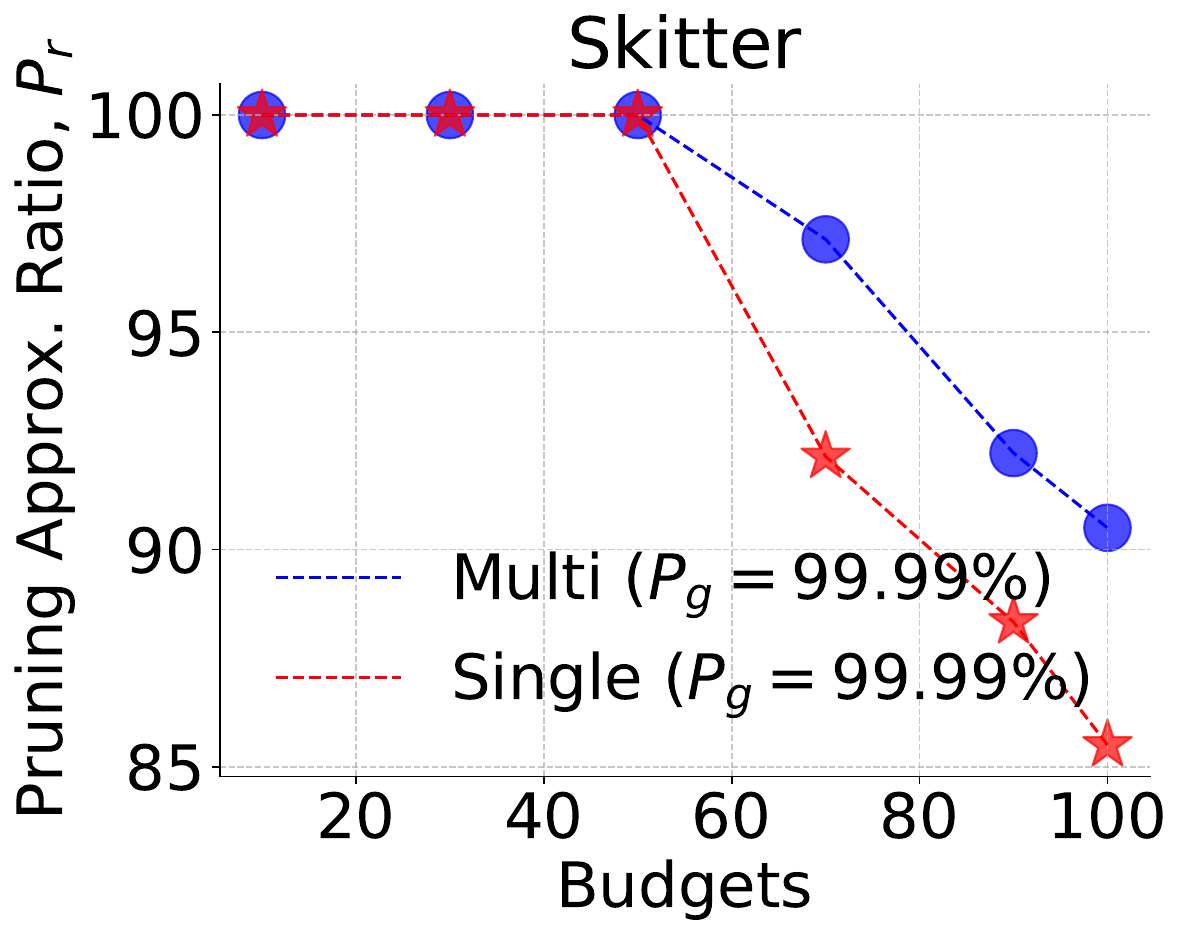}
  }

  \caption{Multi-Budget vs Single Budget for MaxCover}
  \label{}

\end{figure*}

\begin{figure*}[t] 
  \subfigure { 
    \includegraphics[width=0.23\textwidth]{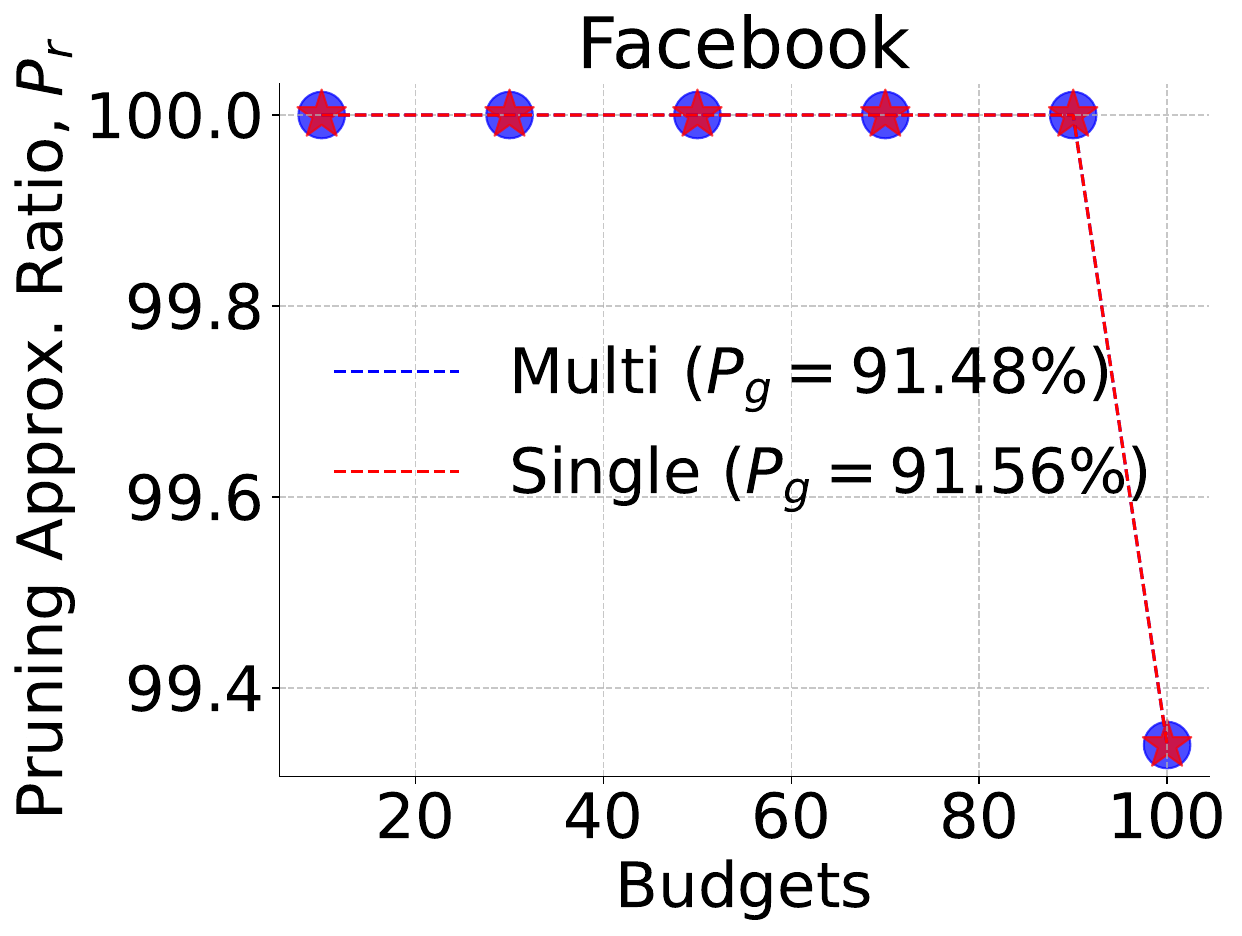}
  }
  \subfigure{ 
    \includegraphics[width=0.23\textwidth]{Figures/IM/Wiki.pdf}
  }
  \subfigure{ 
    \includegraphics[width=0.23\textwidth]{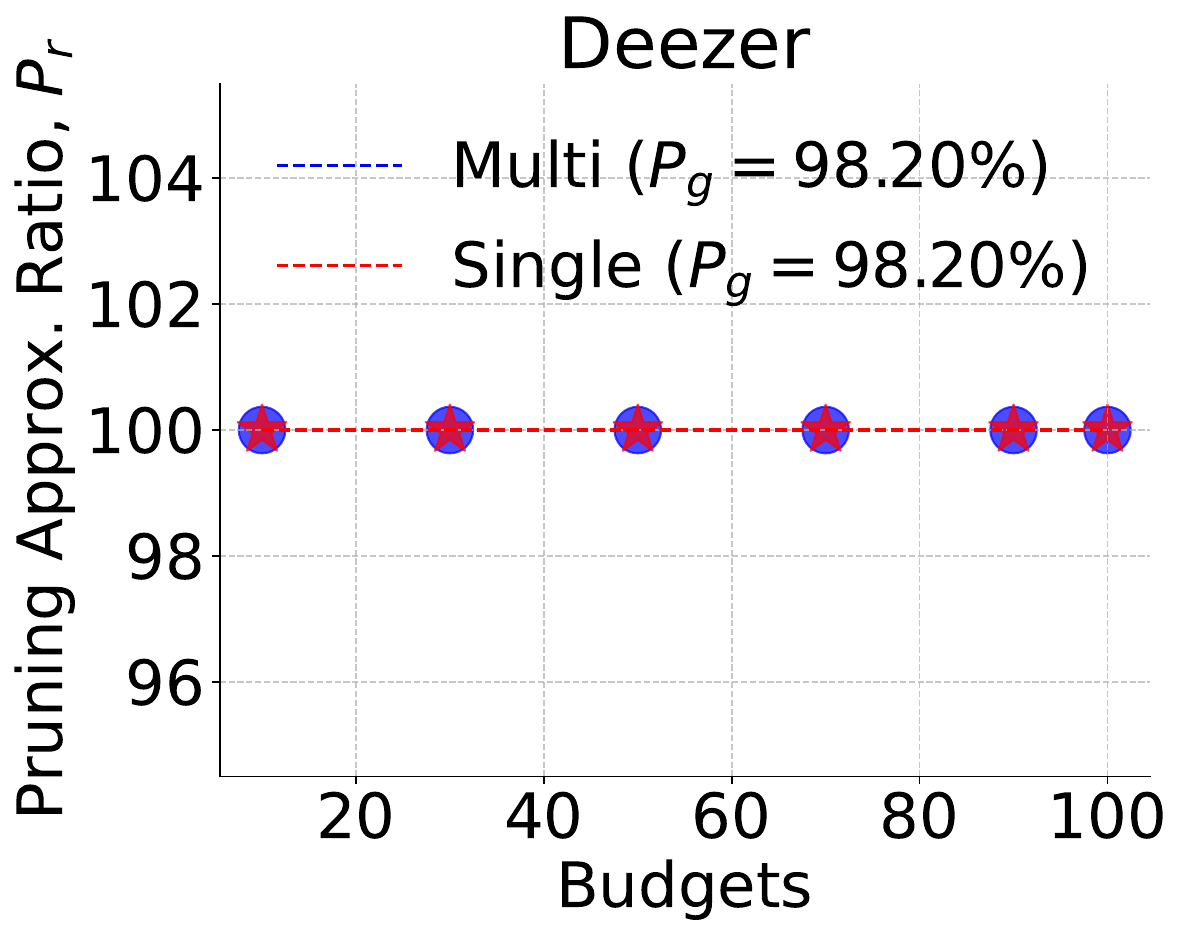}
  }
    \subfigure{ 
    \includegraphics[width=0.23\textwidth]{Figures/IM/Slashdot.pdf}
  }
  \subfigure { 
    \includegraphics[width=0.23\textwidth]{Figures/IM/Twitter.pdf}
  }
  \subfigure{ 
    \includegraphics[width=0.23\textwidth]{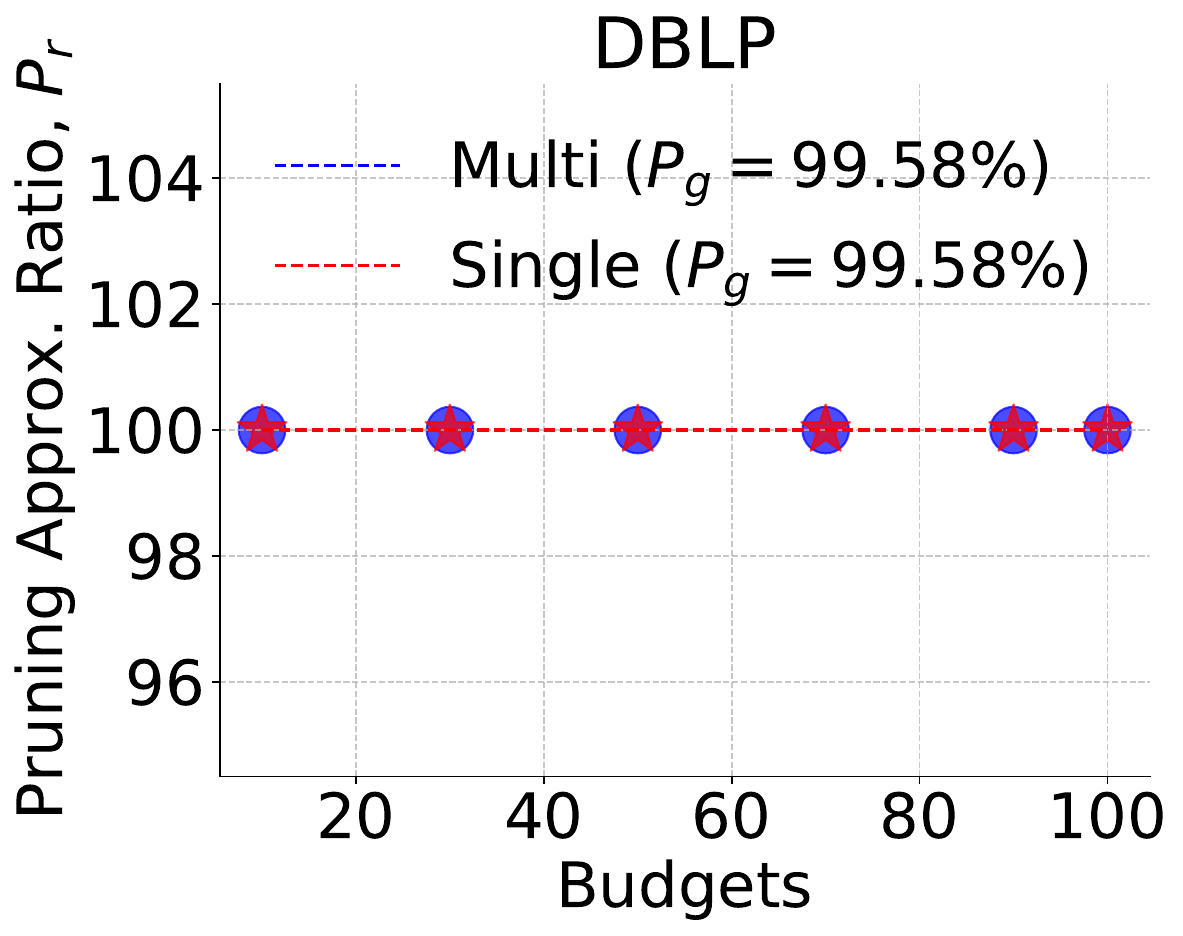}
  }
  \subfigure{ 
    \includegraphics[width=0.23\textwidth]{Figures/IM/YouTube.pdf}
  }
  \subfigure{ 
    \includegraphics[width=0.23\textwidth]{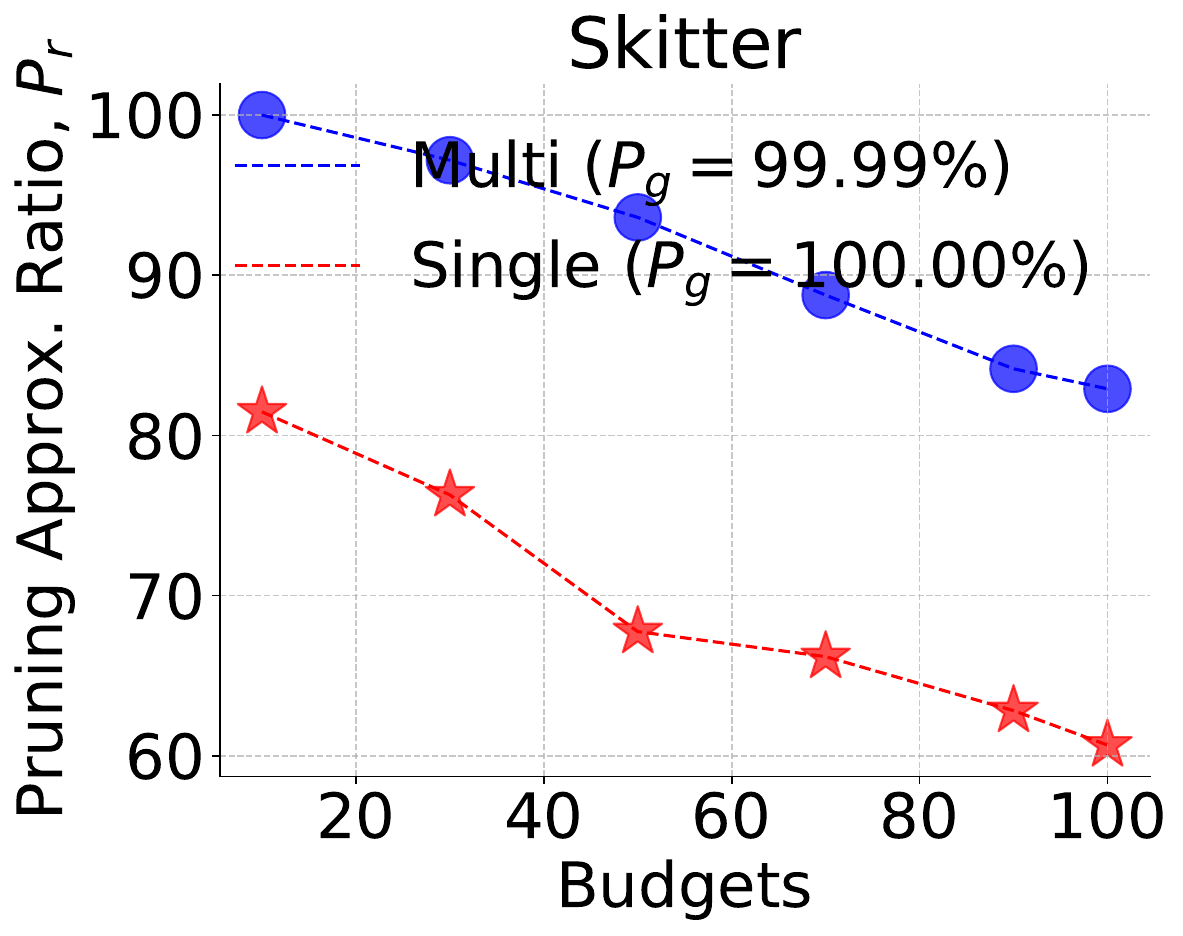}
  }

  \caption{Multi-Budget vs Single Budget for IM}
  \label{}

\end{figure*}

\begin{figure*}[t] 
  \subfigure { 
    \includegraphics[width=0.23\textwidth]{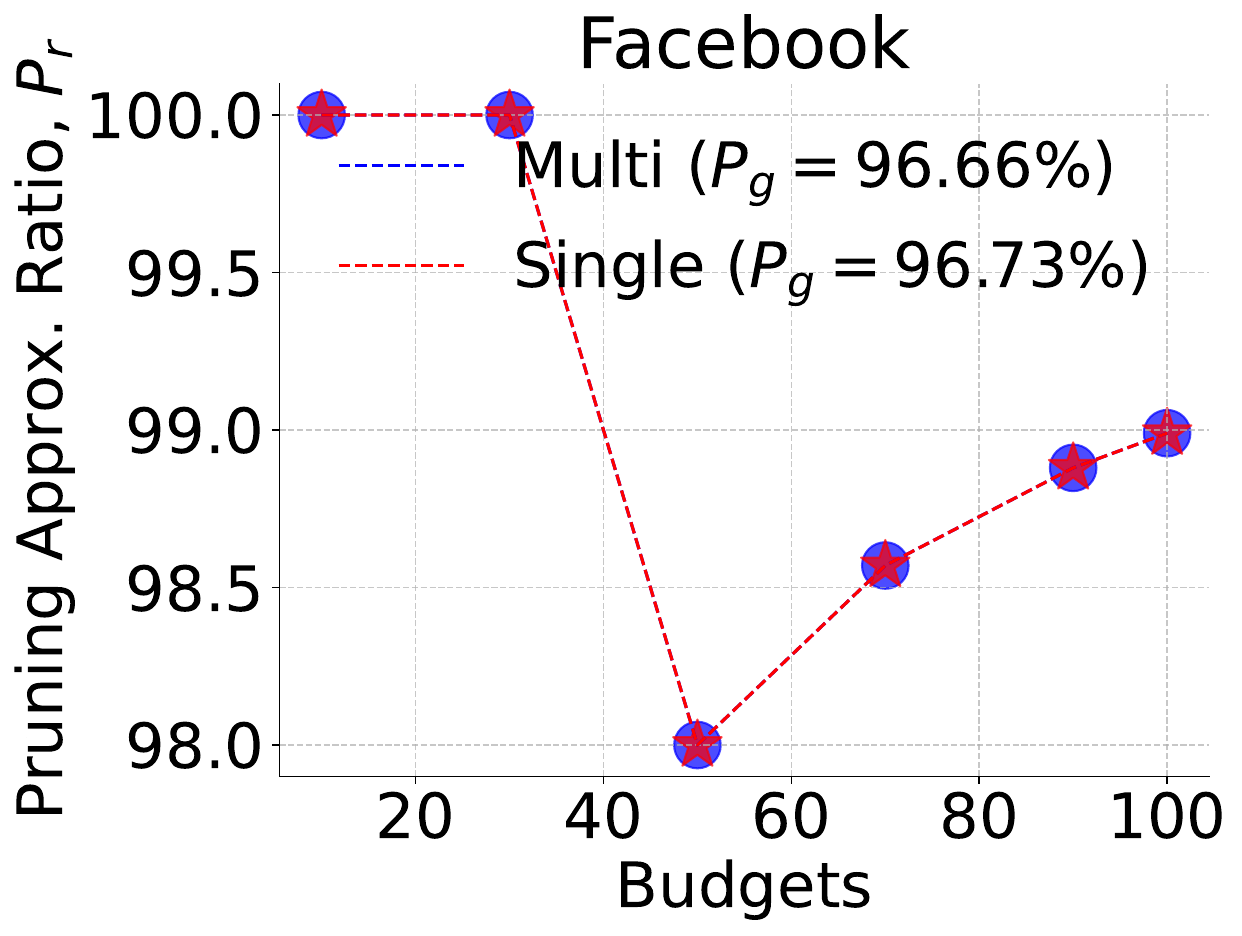}
  }
  \subfigure{ 
    \includegraphics[width=0.23\textwidth]{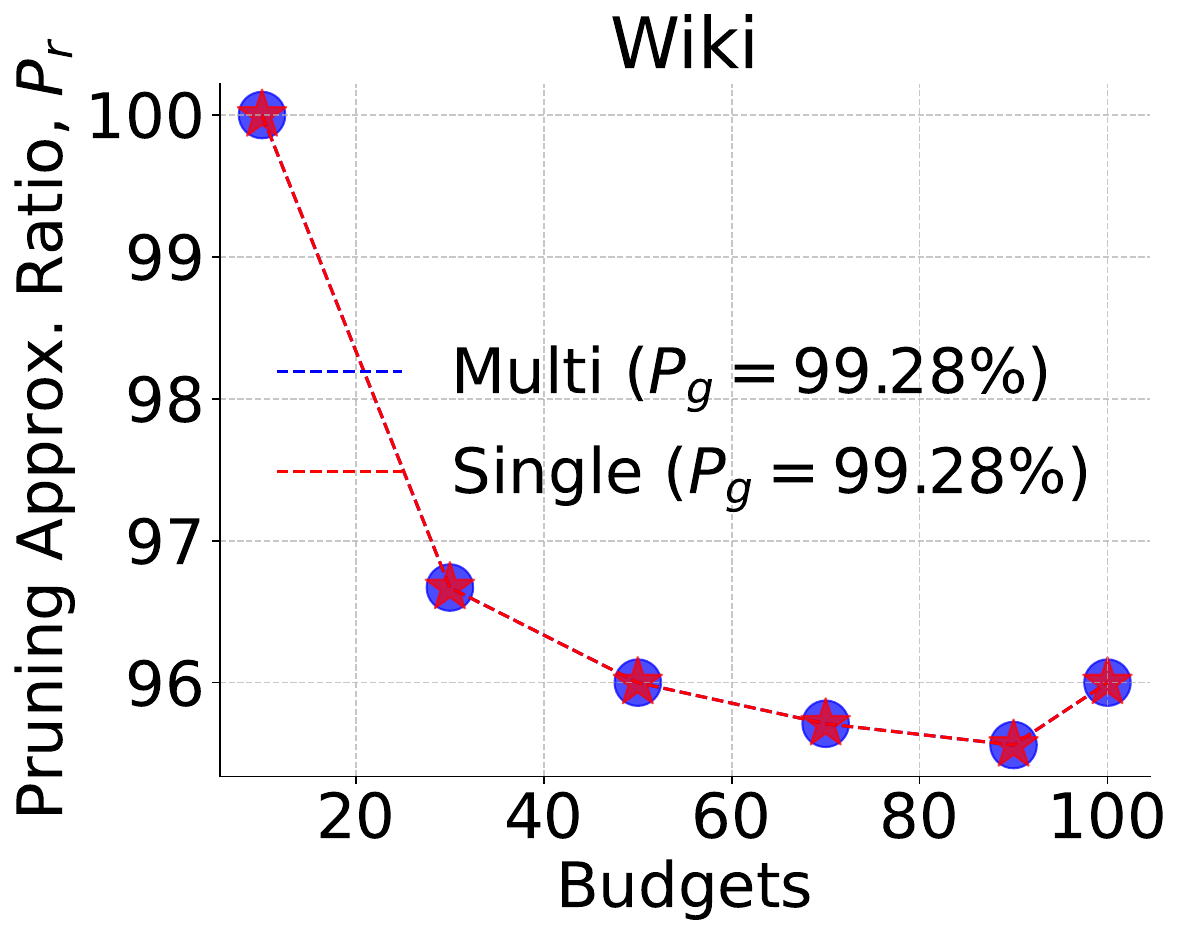}
  }
  \subfigure{ 
    \includegraphics[width=0.23\textwidth]{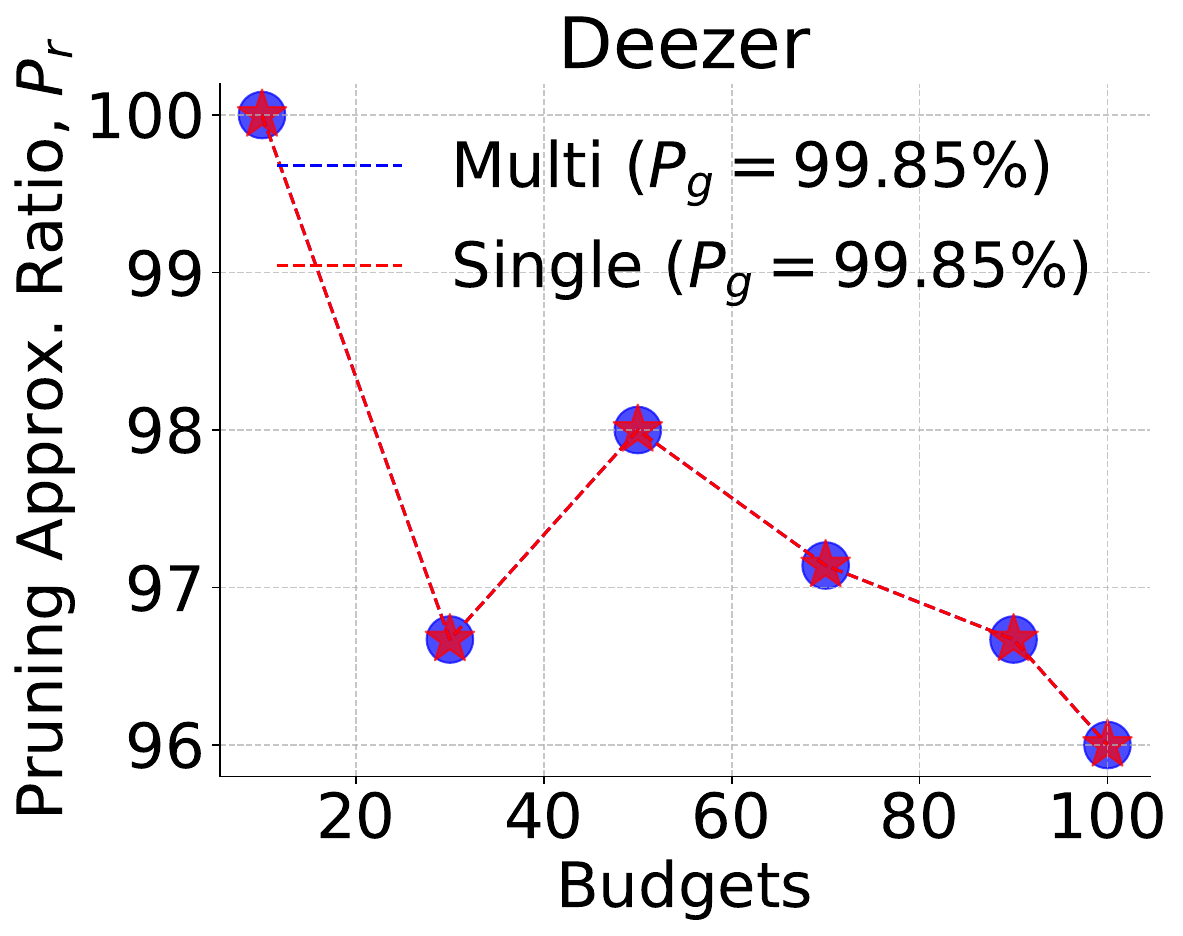}
  }
    \subfigure{ 
    \includegraphics[width=0.23\textwidth]{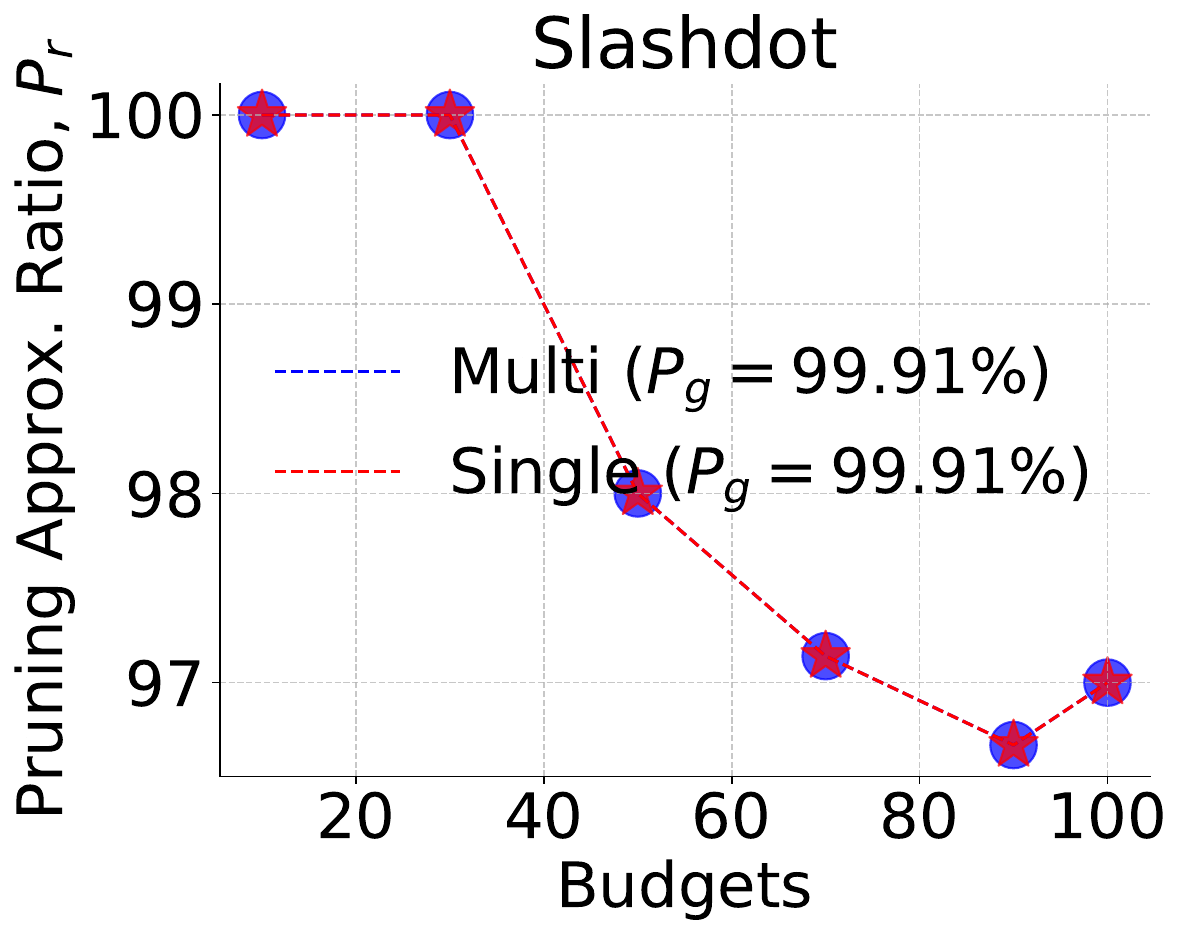}
  }
  \subfigure { 
    \includegraphics[width=0.23\textwidth]{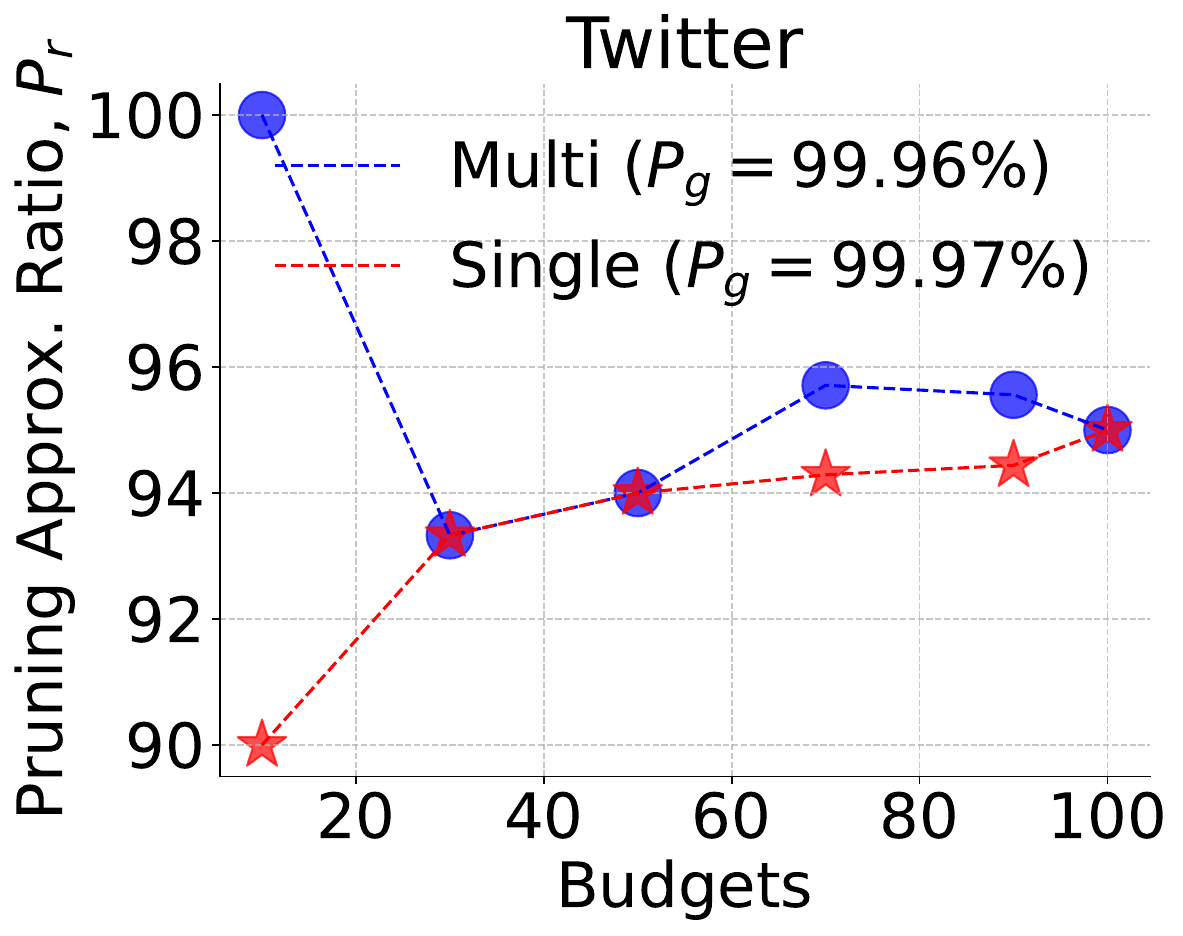}
  }
  \subfigure{ 
    \includegraphics[width=0.23\textwidth]{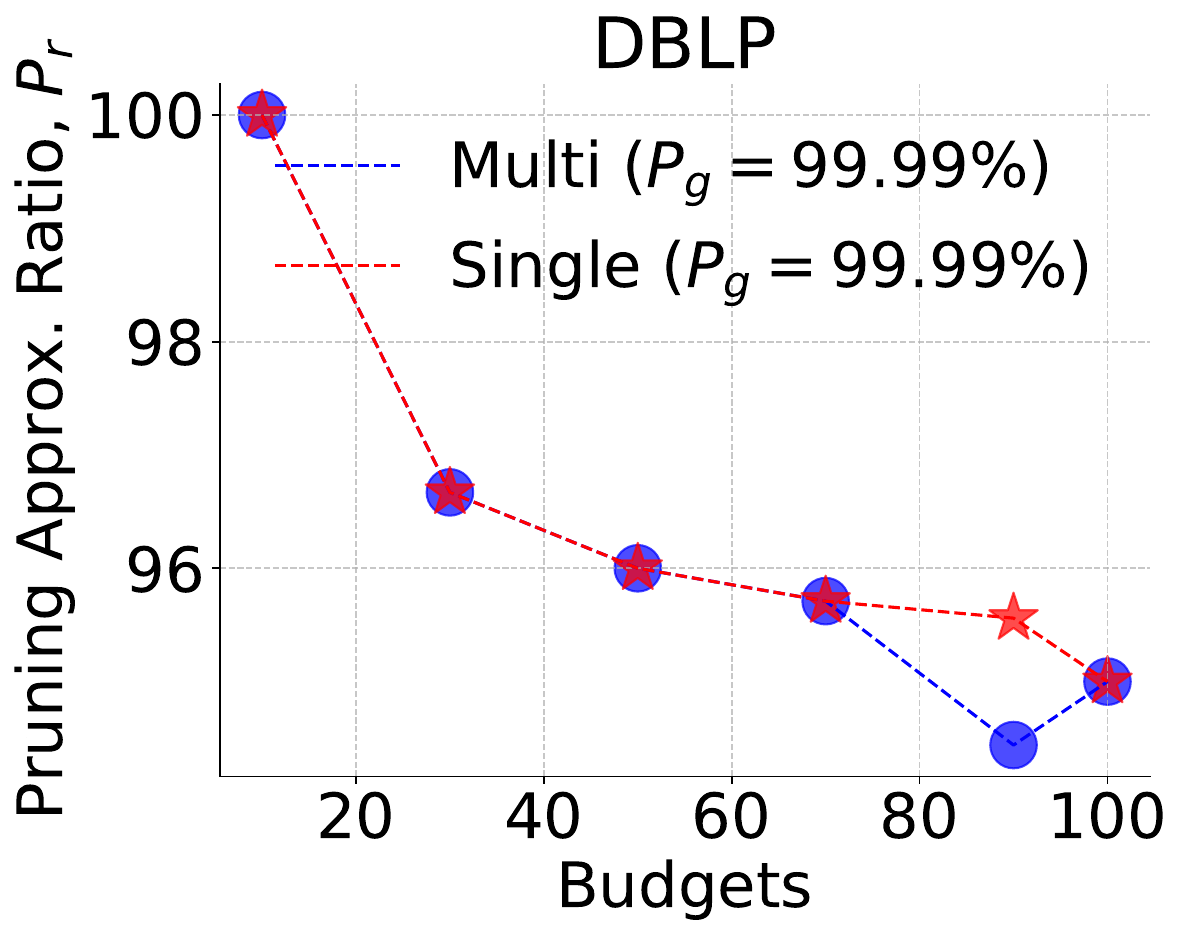}
  }
  \subfigure{ 
    \includegraphics[width=0.23\textwidth]{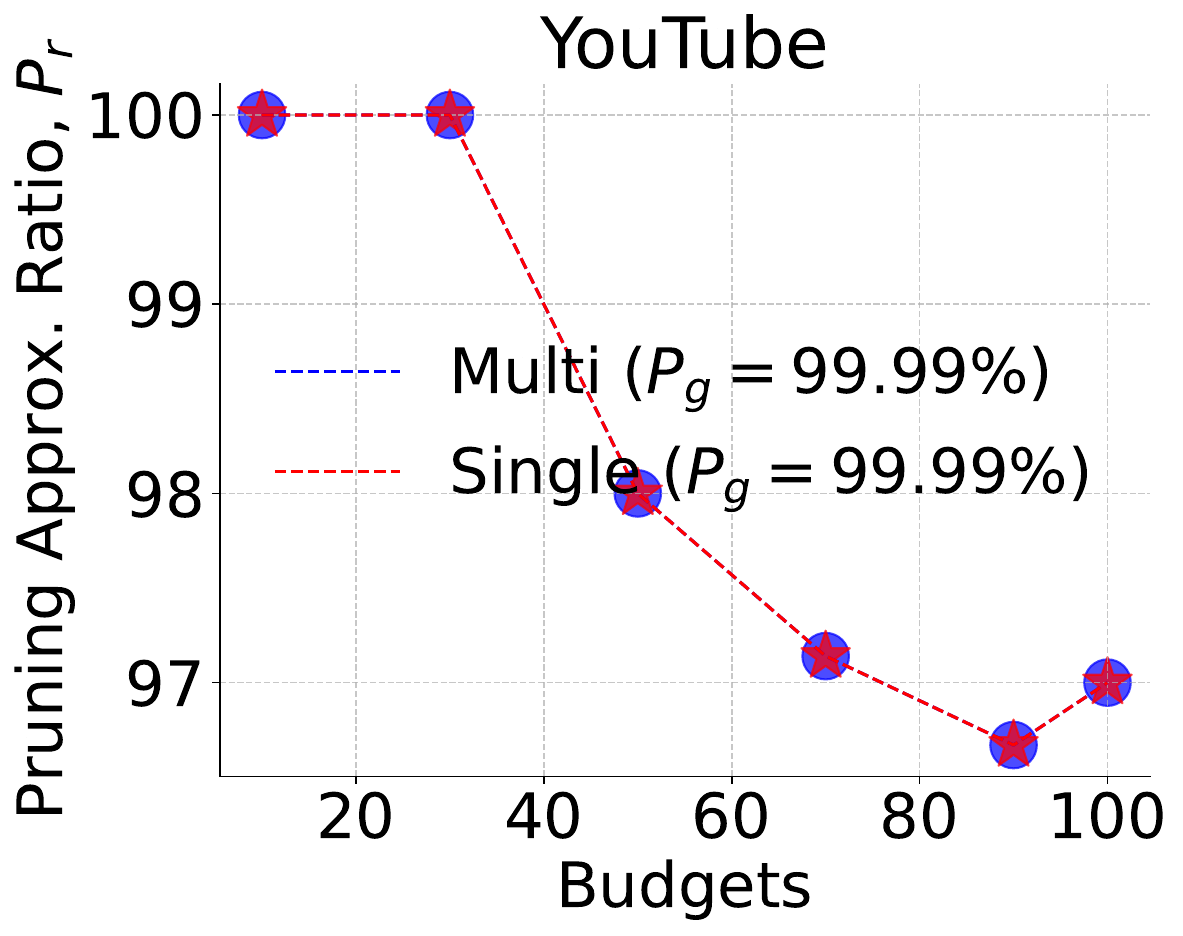}
  }
    \subfigure{ 
    \includegraphics[width=0.23\textwidth]{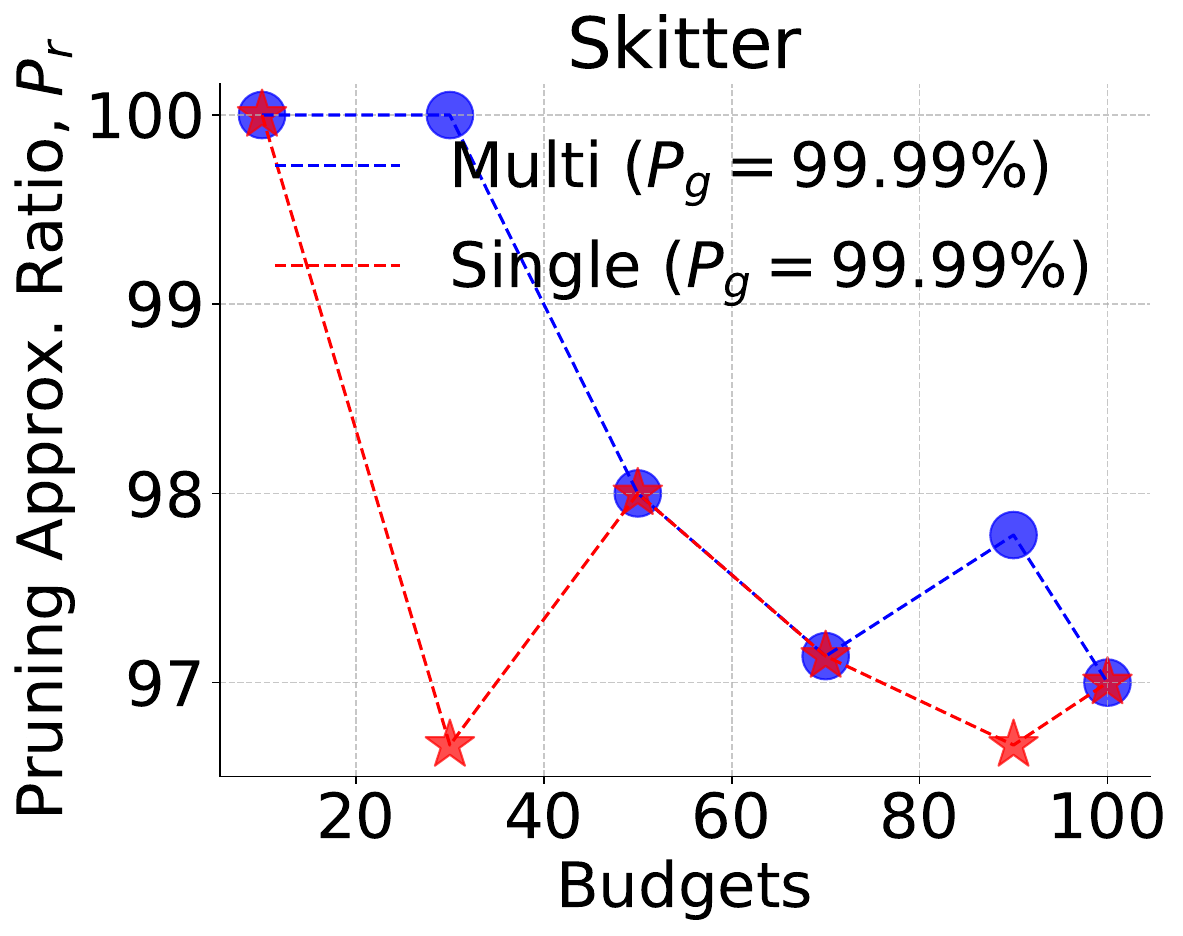}
  }
  
  \caption{Multi-Budget vs Single Budget for MaxCut}
  \label{}

\end{figure*}

\end{document}